\UseRawInputEncoding
\documentclass[]{llncs}
\usepackage{lmodern}
\usepackage{amsmath}
\usepackage{amsfonts}
\usepackage{bm}
\usepackage{afterpage}
\usepackage{xmpincl} 
\usepackage{bbold}
\usepackage{braket}
\usepackage{dsfont}
\usepackage{mathtools}
\usepackage{amssymb}
\usepackage{tikz}
\usetikzlibrary{calc}
\usetikzlibrary{matrix}
\usepackage{verbatim}
\usepackage{tabularx}
\usepackage{amsmath}
\usepackage{array}
\usepackage{longtable}
\usepackage{xspace}
\usepackage[colorlinks=true,
            linkcolor=red,
            citecolor=green,
            urlcolor=red]{hyperref}
\usepackage{xcolor}
\usepackage[capitalise]{cleveref}

\newcommand{\prover}[0]{\mathcal{P}}

\newcommand{\proj}[1]{\ket{#1}\!\!\bra{#1}}

\newcommand{\CompOperatorComp}[1]{\tilde{\mathcal{#1}}}

\newcommand{\sigp}{\ensuremath{\mathsf{\Sigma}\text{-}\mathrm{protocol}}\xspace}

\newcommand{\Sim}{\ensuremath{\mathsf{Sim}}\xspace}
\newcommand{\negl}{\mathrm{negl}}
\newcommand{\comp}{\ensuremath{\mathrm{Comp}}\xspace}
\newcommand{\outputreg}{\ensuremath{O}\xspace}
\newcommand{\Fis}{\ensuremath{\mathsf{Fis}}}
\newcommand{\Ext}{\mathsf{Ext}}
\newcommand{\challspacesize}{\ensuremath{N}}
\newcommand{\N}{\mathbb N}
\newcommand{\instance}{x}

\newif\ifsubmission
\submissiontrue
\submissionfalse

\newif\ifshowcomments
\showcommentsfalse

\ifsubmission
\showcommentsfalse
\fi

\ifshowcomments
\newcommand{\cm}[1]{\textcolor{purple}{\textit{[Chris: #1]}}}

\newcommand{\js}[1]{\textcolor{orange}{\textit{[Jaya: #1]}}}
\else
\newcommand{\cm}[1]{}

\newcommand{\js}[1]{}
\fi

\title{Security of the Fischlin Transform in Quantum Random Oracle Model}
\date{}
\ifsubmission
\author{}
\institute{}
\vspace{-1cm}
\else
\author{Christian Majenz \and Jaya Sharma}
\institute{	Department of Applied Mathematics and Computer Science, Technical University of Denmark \\
	\email{\{chmaj, jaysh\}@dtu.dk} }
\fi

\begin{document}

\maketitle

\begin{abstract}
	The Fischlin transform yields non-interactive zero-knowledge proofs with straight-line extractability in the classical random oracle model. This is done by forcing a prover to generate multiple accepting transcripts through a proof-of-work mechanism. Whether the Fischlin transform is straight-line extractable against quantum adversaries has remained open due to the difficulty of reasoning about the \emph{likelihood} of query transcripts in the quantum-accessible random oracle model (QROM), even when using the compressed oracle methodology. In this work, we prove that the Fischlin transform remains straight-line extractable in the QROM, via an extractor based on the compressed oracle. This establishes the post-quantum security of the Fischlin transform, providing a  post-quantum straight-line extractable NIZK alternative to Pass' transform with smaller proof size. Our techniques are built on different combinations of symmetrization, query amplitude and quantum union bound arguments, as well as tail bounds for sums of independent random variables and for martingales.
\end{abstract}


\ifsubmission
\else
	\setcounter{tocdepth}{2}
	\tableofcontents
\fi

\section{Introduction}

Non-interactive zero-knowledge proofs (NIZKs) obtained from interactive protocols via the Fiat--Shamir paradigm form a central tool in modern cryptography. Many applications require \emph{straight-line extractability} -- the ability to efficiently extract a witness from any succesful prover, without rewinding. The Fischlin transform  \cite{C:Fischlin05} provides a prominent approach for achieving straight-line extractable NIZKs from $\Sigma$-protocols in the random oracle model. It constitutes an alternative to Pass' transform \cite{C:Pass03}, using what would now be called a proof-of-work-like mechanism instead of commitments to enable smaller proofs.

With the advent of quantum computing, it has become necessary to analyze the security of random-oracle-based constructions in the presence of quantum adversaries. This has led to the development of the quantum random oracle model (QROM) \cite{AC:BDFLSZ11}, in which adversaries may query the random oracle in superposition. Over the past decade, a substantial body of work has established a versatile technical toolbox for proving the post-quantum security of classical constructions in the QROM. These techniques have been successfully applied to a wide range of Fiat--Shamir-type transforms, including Pass' transform \cite{TCC:ChiManSpo19,C:DFMS22,C:RotTes25} and related primitives. However, despite its importance, the Fischlin transform has so far resisted  security analysis in the QROM. Existing proof techniques do not readily extend to this setting, and no QROM extractability proof for the Fischlin transform was previously known.

The main difficulty in analyzing the Fischlin transform in the QROM stems from the structure of its extraction mechanism. Informally, a prover succeeds only if it finds an oracle query whose output satisfies a somewhat sparse acceptance condition, which in turn requires many attempts, except with tiny probability. A classical extractor can thus read off two accepting transcripts from the adversarial prover's oracle query inputs, allowing extraction of a witness via special soundness. In the quantum setting, however, this reasoning becomes problematic. The adversary's output depends only on the single accepting query, the unsuccessful ones are in some sense \emph{counterfactual}. This stands in sharp contrast to Pass' transform, where accepting transcripts contain explicit commitments to responses for all challenges that can be leveraged for extraction even in the quantum setting. As a result, proving that a quantum adversary must invest substantial query complexity in order to succeed requires new techniques.

\vspace{.1cm}

\noindent \textbf{Our results.}
In this work, we prove the straight-line extractability of the Fischlin transform in the QROM. Our result provides concrete bounds for a wide range of parameter choices that include an infinite family of parameter sets for asymptotic security. In more detail, we have the following
\begin{theorem}[Informal version of \cref{cor:main}]
	Let $\Sigma$ be a \sigp for a witness relation $R$ with special soundness and unique responses. Under suitable conditions on the parameters of the Fischlin transform that allow unbounded values for the parallel repetition parameter of the Fischlin transform, $k$, (and thus for the security paramter,) the Fischlin transform $\mathsf{Fis}[\Sigma]$ is a proof of knowledge with straight-line extractability in QROM. The simulation can be made perfect and the extractor succeeds whenever a $q$-query malicious prover does, except with probability no more than
	\begin{align*}
		\varepsilon_{\mathrm{ex}}\le q^2\cdot \negl(k).
	\end{align*}
\end{theorem}

\vspace{.1cm}

\noindent \textbf{Technical Overview.}
Defining the straight-line extractor for the Fischlin transform is straightforward. We simulate the random oracle using a compressed oracle \cite{C:Zhandry19}. Once the prover has finished and its proof is verified, the extractor measures the compressed oracle database register. If the measurement result contains two accepting transcripts for the same repetition with different challenges, special soundness of the \sigp is used to produce a witness.

Proving that this extractor works turns out to be a formidable challenge. Known techniques for straightline extraction in the QROM are not sufficient. There seem to be two such techniques in the literature. Firstly, in \cite{EC:Unruh15}, straightline-extraction for commitments is achieved by simulating a length- preserving quantum-accessible random oracle (QRO) in a way that allows finding pre-images. And secondly, \cite{TCC:ChiManSpo19} showed how to use the compressed oracle for straightline extraction. This mechanism was further developed \newline\cite{EC:DFMS22,C:DFMS22,C:RotTes25} and applied \cite{AC:AHJMRY23,C:HJMN24,C:BBBDKM25}. None of these techniques can be used for the Fischlin transform though -- the absence of commitments prevents straightforward generalization, and the fact that extractor failure is characterized by the absence of certain patterns in the compressed oracle database, rather than their existence, prevents transition capacity reasoning \cite{EC:CFHL21}.

The main work goes into proving that no prover can succeed with significant probability at outputting a proof that involves a fixed commitment vector chosen ahead of time (``deterministic prover'') so that the extractor still fails. This result is then lifted to a general prover using an intricate argument based on quantum query amplitudes and a quantum union bound.

To show the result for deterministic provers, our strategy is reminiscent of the Monty Hall paradox. Here, a game master presents three doors. Behind a random one is a car, while there are goats behind the other two. The participant can now choose a door, which is, however, not opened. The game master now opens a \emph{different} door which reveals a goat. The participant can now decide: stick to their choice of door, or choose the only remaining door instead. The optimal strategy is to switch, which reveals the car with probability 2/3. What makes this situation appear paradoxical is the fallacy of regarding the remaining door as a fixed object instead of a random variable. The game master chooses the door they open depending on the randomness of the placement of the car and the participant's choice, which correlates the latter with the car's placement.

When a prover outputs a valid Fischlin proof, a number of hash values are now known (to fulfill the PoW condition, in the following discussion simplified to ``output 0''). This can only be achieved by correlating the choice of challenges with the random oracle. Clearly this, in turn, necessitates that some hash outputs not involved in proof verification are now \emph{less} likely to be 0 than suggested by the uniform distribution, otherwise the oracle output distribution would be biased. In the compressed oracle, analogously, the oracle registers corresponding to valid transcripts not involved in the verification of the proof output by the adversarial prover cannot contain the uniform superposition of outputs anymore. This deviation is then exploited by the extractor.

From a technical perspective, the task of making this argument rigorous is formidable. The output length of the random oracle is logarithmic, preventing the usual tight relationship between oracle register content and adversarial view up to negligible error. We therefore need to explicitly analyze every switch between the ``compressed'' and ``uncompressed'' bases of the superposition oracle using Chernoff bounds. The high-level strategy is as follows. Pick a number $\eta$ such that (i) it is unlikely that there are $\eta$ oracle inputs in a certain sub-domain that map to zero (by a standard Chernoff bound), but (ii) it is likely to find roughly $\eta-k$ zeroes, where $k$ is the number of parallel repetitions in the Fischlin transform and thus the number of hash values that need to be zero for a Fischlin proof to be accepted. We then use (a somewhat cumbersome compressed oracle version of) the fact that it is unlikely to find $\eta$ zeroes, but likely to find $\eta$ zeroes given that the prover has succeeded and the extractor has failed, to conclude that the probability that the prover succeeds and the extractor fails is small.
Proving that it is likely to find $\eta$ zeroes given that the prover has succeeded and the extractor has failed is non-trivial. Conditioning (or rather projecting) on these events renders the oracle register non-uniform and concentration inequalities for independent Bernoulli random variables no longer apply. Fortunately, we can simplify the analysis by symmetrizing the state with respect to permutations of the oracle database registers not associated to a hash input queried during verification. We can then define a martingale based on the outcomes of sequentially measuring these database registers. As a result of the symmetrization, the conditional expectations of these measurements can be bounded. This bound translates into a bound on the expected number of zeroes, and an application of the Azuma-Hoeffding concentration inequality supplies the necessary tail bound for deviations from that expected value.

Overall, our analysis establishes that, even in the presence of quantum superposition queries, success in the Fischlin transform implicitly certifies the existence of some ``hidden'' accepting transcripts. While these transcripts may not be classically observable, they manifest themselves in the compressed oracle database, enabling straight-line extraction in the QROM.

\vspace{.1cm}

\noindent \textbf{Organization.}
Section~2 reviews the QROM and the compressed oracle framework. Section~3 proves straight-line extractability of the Fischlin transform in the QROM. Section~4 establishes zero-knowledge.

\section{Preliminaries}

For asymptotic security analysis, protocol parameters are expressed as a function of the security parameter $\lambda\in \mathbb N$. A (possibly negative) function $f(\lambda)$ is called negligible if $|f(\lambda)| \leq 1 / p(\lambda)$ for any polynomial $p(\lambda)$ and all sufficiently large $\lambda$. We use $\negl$ to denote a negligible function.  A function which is not negligible is called noticeable. For any witness relation $R\subset \mathcal X\times\mathcal W$ we call the elements of $\mathcal X$ instances, the elements of $\mathcal W$ witnesses, and define the associated language $\mathcal L_R=\{x\in\mathcal X|\exists w\in\mathcal W: (x,w)\in R\}$. We denote by $\delta(\cdot,\cdot )$ the maximal 1-sample computational distinguishing advantage between two probability distributions. We allow algorithms as arguments here, which indicate computational distinguishability of their outputs.

A common idealization is the random oracle model where parties have access to a random function $H$ with some domain and range depending on $\lambda$. Our random oracles have output length at most polynomial in the security parameter.

\begin{definition}[$\mathsf{\Sigma}$-Protocol]\label{sigma}
	A \sigp $\Sigma = (\prover,\mathcal{V})$ for a relation
	$R \subseteq \mathcal{X} \times \mathcal{W}$ is a three-round two-party interactive protocol between a Prover $\prover=(\prover_1,\prover_2)$ and a verifier $\mathcal{V}$ where the prover produces $(a,\mathsf{st})\gets\prover_1(x,w)$ and sends the commitment $a\in \mathcal M$ to the verifier, the verifier sends a random challenge $c\leftarrow \mathcal C$ (\emph{public-coin}) and the prover sends a response $z\gets \prover_2(\mathsf{st},c)$. Finally, the verifier outputs the verdict $b\gets \mathcal V(x,a,c,z)$, with $b=1$ meaning ``valid''.  We denote the part of the interaction that produces the communication transcript as $(a, c, z) \leftarrow\langle\prover(x, w), \mathcal{V}\rangle$. It has to satisfy the following properties:
	\begin{enumerate}
		\item Completeness. For any $(x, w) \in R$, any $(a, c, z) \leftarrow\langle\prover(x, w), \mathcal{V}\rangle$ it holds $\mathcal{V}(x, a,  c, z)=1$.

		\item Special Soundness. There exists a probabilistic polynomial-time (PPT) extractor algorithm $\Ext_{\mathrm{ss}}$ such that for all  $x \in \mathcal{X}$ and for all pairs of transcripts $(a, c, z)$ and  $(a,c^{\prime}, z^{\prime})$ with $c \neq c^{\prime}$  such that $\mathcal{V}(x, a, c, z)=\mathcal{V}\left(x, a, c^{\prime}, z^{\prime}\right)=1$, the extractor produces a witness $w \leftarrow \Ext_{\mathrm{ss}}\left(x, a, c, z, c^{\prime}, z^{\prime}\right)$ with $(x, w)\in R$.

	\end{enumerate}
\end{definition}
\begin{definition}[Min-Entropy]
	The \textbf{min-entropy} of a discrete random variable \( X \) over a finite set \( \mathcal{X} \) is defined as:
	\[
		H_{\infty}(X) = -\log_2 \left( \max_{x \in \mathcal{X}} \Pr[X = x] \right)
	\]
	where \( \Pr[X = x] \) denotes the probability that \( X \) takes the value \( x \).
\end{definition}

We now list some additional properties required by our \sigp the make the transformation from interactive to non-interactive, work:

\begin{definition}[Commitment Entropy] Let $\Sigma$ be a \sigp. We say $\Sigma$ has \emph{commitment entropy} $\alpha\in\mathbb R_{\ge 0}$ if , for any $(x, w) \in R$ and $a\leftarrow \prover_{1}(x, w)$ it holds that $H_\infty(a)\ge \alpha$.
\end{definition}

\begin{definition}[Unique Responses]\label{Unique Responses} Let $\Sigma$ be a \sigp . We say $\Sigma$ has \emph{unique responses}, if for all tuples $\left(x,a,c,z,z'\right) \leftarrow A(\lambda)$ with $z\neq z'$, either $\mathcal V(x,a,c,z)=0$ or $\mathcal V(x, a,c,z')=0$. \end{definition}

\begin{definition}[Honest-Verifier Zero-Knowledge (HVZK)]  Let $\Sigma$ be a \sigp for relation $R$. We say $\Sigma$ has \emph{honest-verifier zero-knowledge (HVZK)} if there exists a PPT algorithm $\Sim_{\mathrm{HVZK}}$, the zero-knowledge simulator, such that for any quantum polynomial-time (QPT) algorithm $\mathcal{D}=\left(\mathcal{D}_0, \mathcal{D}_1\right)$, it holds that
	\begin{align*}
		 & \Big|\Pr_{\substack{(x, w,\mathsf{st})  \leftarrow \mathcal{D}_0\left(1^\lambda\right)  \\ (a,c,z) \leftarrow\langle\prover(x, w), \mathcal{V}\rangle}}[1\gets \mathcal{D}_1(a, c, z,\mathsf{st})\wedge (x,w)\in R]\\
		 & \quad-\Pr_{\substack{(x, w,\mathsf{st})  \leftarrow \mathcal{D}_0\left(1^\lambda\right) \\ (a,c,z) \leftarrow\Sim_{\mathrm{HVZK}}(x)}}[1\gets \mathcal{D}_1( a, c, z,\mathsf{st})\wedge (x,w)\in R]\Big|\le \negl(\lambda).
	\end{align*}

\end{definition}

Below we sometimes use a stronger zero-knowledge property where the simulator is able to generate proofs for a specific challenge.

\begin{definition}[Special Honest Verifier Zero-Knowledge]\label{SHVZK DEFN}
	A \sigp $\Sigma$ is special honest-verifier zero-knowledge (SHVZK) if it is HVZK and there exists a simulator $\Sim_{\mathrm{SHVZK}}$ such that $\Sim_{\mathrm{HVZK}}(x)$ simulates transcripts $(a,c,z)$ by sampling $c\gets\mathcal C$ and running $(a,z)\gets\Sim_{\mathrm{SHVZK}}(x,c)$.

\end{definition}

We now formally describe a Non-Interactive Zero-Knowledge(NIZK) proof with online extractors in the Random Oracle Model.
\begin{definition}[Non-Interactive Zero-Knowledge proof (NIZK)]\label{Nizk defn}
	A pair $(\prover, \mathcal{V})$ of probabilistic polynomial-time algorithms is called a non-interactive zero-knowledge proof (NIZK) for relation $R$ if the following holds:

	\noindent   \textbf{Completeness.} For any $(x, w) \in R$ and any $\pi \leftarrow \prover^H(x, w)$ we have \\ $\Pr\left[\mathcal{V}^H(x, \pi)=1\right]= 1-\negl(\lambda)$.

	\noindent\textbf{Zero-Knowledge.}
	There exists a probabilistic polynomial-time simulator $\Sim$ such that for any
	quantum polynomial-time distinguisher
	$\mathcal D=(\mathcal D_0,\mathcal D_1)$, we have
	\begin{align*}
		\Big|
		 & \Pr\Big[
			\mathcal D_1^{H}(\pi,\mathsf{st})=1 :
			(x,w,\mathsf{st})\leftarrow \mathcal D_0^{H}(1^\lambda),\,
			(x,w)\in R,\,
			\pi\leftarrow \prover^{H}(x,w)
			\Big]
		\\
		-
		 & \Pr\!\Big[
			\mathcal D_1^{H'}\!\!(\pi,\mathsf{st})\!=\!1\! :\!
			(x,w,\mathsf{st})\!\leftarrow \!\mathcal D_0^{H}(1^\lambda),
			(x,w)\in R,
			(H',\pi)\!\leftarrow \!\Sim(x)
			\Big]
		\Big|
		\!\le\! \negl(\lambda).
	\end{align*}
	Here, the simulator outputting $H'$ is to be understood as outputting (reprogramming) instructions how to instantiate $H'$ relative to $H$.

\end{definition}

In the QROM, straight-line extractability can be defined as follows, following \cite{C:DFMS22}. Let ${\cal P}^*$ be a dishonest prover that potentially outputs some additional auxiliary (possibly quantum) output $Z$ next to $\pi$.
We then consider an interactive algorithm $\Ext$, called {\em online extractor}, which takes the security parameter $\lambda$ as input and simulates the answers to the oracle queries in the execution of ${\cal V}^H \circ {\cal P}^*{}^H(x)$, which we define to run  $(\pi,Z) \leftarrow {{\cal P}^*}{}^H(x)$ followed by $v \leftarrow {\cal V}^H(\instance,\pi)$; furthermore, at the end, $\Ext$ outputs $w \in {\cal W}$.
We denote the execution of ${\cal V}^H \circ {\cal P}^*{}^H(x)$ with the calls to $H$ simulated by $\Ext$, and considering $\Ext$'s final output $w$ as well, as $(\pi,Z;v;w) \leftarrow {\cal V}^{\Ext} \circ {\cal P}^*{}^{\Ext}(x)$.
\begin{definition}[Definition 3.1 in \cite{C:DFMS22} specialized to non-adaptive adversaries]\label{def:OnlineExtr}
	A NIZK in the (QROM) for relation $R\subset\mathcal X\times \mathcal W$ is a {\em proof of knowledge with straight-line extractability} against non-adaptive adversaries if there exists an online extractor $\Ext$, and functions $\varepsilon_\text{\rm sim}$ (the {\em simulation error}) and $\varepsilon_\text{\rm ex}$ (the {\em extraction error}), with the following properties. For any $\lambda \in \N$, for any $x\in\mathcal X$ and for any dishonest prover ${\cal P}^*$ making no more than~$q$ queries,
	{
			$$
				\delta\bigl( [(\pi,Z,v)]_{{\cal V}^H \circ {\cal P}^*{}^H(x)} , [(\pi,Z,v)]_{{\cal V}^{\Ext} \circ {\cal P}^*{}^{\Ext}(x)} \bigr) \leq \varepsilon_\text{\rm sim}(\lambda,q)
			$$
			and
			$$
				\Pr\bigl[ v = {\tt accept} \,\wedge\, (\instance,w) \not\in R: (\pi,Z;v;w) \leftarrow {\cal V}^{\Ext} \circ {\cal P}^*{}^{\Ext}(x) \bigr] \leq \varepsilon_\text{\rm ex}(\lambda,q) \, .
			$$
			Furthermore,
			the runtime of $\Ext$ is polynomial in $\lambda+q$, and $\varepsilon_\text{\rm sim}(\lambda,q)$ and $\varepsilon_\text{\rm ex}(\lambda,q)$ are negligible in $\lambda$ whenever $q$ is polynomial in $\lambda$. }
	\label{def:PoKOnline}\end{definition}

\textbf{The Fischlin transform.}

We use the version of the Fischlin transform described in \cite{AC:Konshe22}.

\begin{definition}[The Fischlin transform]\label{Fischlin Protocol}
	The security parameter $\lambda$ defines the integers $k, \ell, t$,  related as $k\cdot \ell=\lambda$ and $t=\lceil\log \lambda\rceil \cdot \ell$. The protocol uses a hash function $H:\mathbf{X}=\{(\mathbf{a},i,c_i,z_i)\}=X_{\mathbf{a}}\times X_i\times X_{c_i}\times X_{z_i} \to\mathcal{Y}=\{0,1\}^{\ell}$, modeled as a random oracle, and a \sigp  $\Sigma=\left(\left(\prover_{\Sigma,1}, \prover_{\Sigma,2}\right), \mathcal{V}_{\Sigma}\right)$. The prover and verifier of the Fischlin transform $\Fis[\Sigma]$ are defined as follows.

	\noindent$\prover^{\Fis[\Sigma]}(x, w):$ \\
	1. For each $i \in[k]$, compute $\left(a_i\right.$, $\left.\mathrm{st}_i\right) \leftarrow \prover_{\Sigma,1}(x, w)$\\
	2. Set $\boldsymbol{a}=\left(a_i\right)_{i \in[k]}$, and initialize $c_i=-1$ for each $i \in[k]$\\
	3. For each $i \in[k]$, do the following:\\
	(a) If $c_i>t$, abort. Otherwise increment $c_i$ and compute $z_i=\prover_{\Sigma,2}\left(\operatorname{st}_i, c_i\right)$\\
	(b) If $H\left(\boldsymbol{a}, i, c_i, z_i\right) \neq 0^{\ell}$, repeat Step 3a\\
	4. Output $\pi=(\mathbf a, \mathbf c, \mathbf z)$, where $\mathbf c=\left(c_i\right)_{i \in[k]}$ and $\mathbf z$ are defined analogously.

	\noindent$\mathcal{V}^{\Fis[\Sigma]}(x, \pi):$\\
	1. Parse $(\mathbf a, \mathbf c, \mathbf z)=\pi$\\
	2. For each $i \in[k]$, verify that $H\left(\boldsymbol{a}, i, c_i, z_i\right)=0^{\ell}$ and $\mathcal{V}_{\Sigma}\left(x,\left(a_i, c_i, z_i\right)\right)=1$, aborting with output 0 if not\\
	3. Accept by outputting 1
\end{definition}
For simplicity, we set the parameter $t$ equal to the size of the challenge space of the underlying \sigp for the remainder of this article.

\subsection{Quantum Random Oracle (QROM) and Compressed Oracle Technique}\label{subsection: QROM}

In this section we summarize the mathematical tools and notational conventions used throughout the paper.
We review basic facts about Hilbert spaces and operators, introduce the computational and Fourier bases for finite Abelian groups, and recall the standard encoding of functions into quantum oracles.
These foundations are needed to formalize the compressed oracle technique and to analyze the behavior of quantum adversaries making oracle queries.

\textbf{Hilbert Spaces and Operators.}
We work over finite-dimensional complex Hilbert spaces. Unless stated otherwise, $\mathcal{H}=\mathbb{C}^d$ for some $d$, with the usual bra–ket notation. Linear operators between spaces $\mathcal{H}$ and $\mathcal{H}'$ are denoted $\mathcal{L}(\mathcal{H},\mathcal{H}')$, with $\mathcal{L}(\mathcal{H})$ for endomorphisms.
A (pure) quantum state is a unit vector $|\psi\rangle\in\mathcal{H}$.

\textbf{Norms.}
For $A\in\mathcal{L}(\mathcal{H},\mathcal{H}')$, $\|A\|$ denotes the operator norm. If $\mathcal{H}=\bigoplus_{i=1}^m \mathcal{H}_i$ and $A$ acts as $B_i$ on each $\mathcal{H}_i$, then $\|A\|=\max_i\|B_i\|$. When $A$ is defined only on a subspace, it can be extended by zero outside without ambiguity in its norm.

Any finite set  $\mathcal{Y}$ defines an associated quantum register with state space $\mathbb{C}\mathcal Y$ computational basis  $\{|y\rangle\}_{y\in\mathcal{Y}}$.
We will consider the extension $\overline{\mathcal{Y}}=\mathcal{Y} \cup\{\perp\}$
and the superspace $\mathbb C\overline{\mathcal{Y}}\supset \mathbb{C}\mathcal Y$.

We will denote the function table stored in a quantum register by $|H\rangle=\bigotimes_x|H(x)\rangle\in \left(\mathbb C\mathcal Y\right)^{\otimes \mathcal X}$.We implicitly consider the different registers to be labeled by $x \in \mathcal{X}$ in the obvious way.

\textbf{Compressed Oracle Technique.} The compressed oracle formalism introduced by Zhandry \cite{C:Zhandry19} is a way to simulate a quantum-accessible random oracle in a way that provides additional features for reductions. Instead of storing the full random function, we maintain a \emph{compressed database} $D$ that records only the points queried so far and their outputs. This suffices for simulating any adversary's interaction while simplifying the analysis.

Instead of a uniformly random choice of function, consider a superposition $\frac{1}{\sqrt{|\mathfrak H|}}\sum_{H}|H\rangle $, where $\mathfrak{H}$ is the set of functions $H$ is sampled uniformly from. This is the purified oracle which is indistinguishable from the original random oracle for any (quantum) query algorithm since the queries commute with measuring the superposition. The initial state of the oracle registers is given by

\begin{equation}\label{oracle init state eq}
	\left|\Pi_{0}\right\rangle=\frac{1}{\sqrt{|\mathfrak H|}}\sum_H|H\rangle=\frac{1}{\sqrt{|\mathfrak H|}}\bigotimes_{x}\left(\sum_{y}|y\rangle_{D_x}\right)\eqqcolon\bigotimes_{x}|+_{\mathcal Y}\rangle_{D_x}.\end{equation}

\textbf{Compressed Oracle.} The Compressed Oracle\footnote{More precisely we are using the \emph{pre-compressed} oracle, the compressed oracle is obtained from the pre-compressed oracle via the standard sparse representation} is now obtained by re-defining each local register $D_x$ to have state space $\mathbb C\overline{\mathcal{Y}}$, embedding the superposition oracle into it, and applying the basis change unitary

\begin{equation}\label{Comp}
	\comp=|\perp\rangle\langle+^{\ell}|+|+^{\ell}\rangle\langle\perp|+(\mathbb{1}-|+^{\ell}\rangle\langle+^{\ell}|-|\perp\rangle\langle\perp|).
\end{equation}
This compression operator acts on a single subregister of the compressed oracle database register. As for any operator, we denote by $\comp_{D_x}$ the compression operator applied to a particular register $D_x$. For the compression operator, we additionally abuse notation slightly by writing
\begin{align*}
	\comp_{D_S}=\bigotimes_{x\in S}\comp_{D_x}
\end{align*}
and thus also $\comp_D=\left(\comp^{\otimes |\mathcal X|}\right)_D$.

The compression operator Comp${}_D:=\bigotimes_{x}$ Comp $_{D_x}$ maps $\left|\Pi_{0}\right\rangle$ to

\begin{align*}
	\left|\Delta_{0}\right\rangle & :=\comp\left|\Pi_{0}\right\rangle=\left(\bigotimes_{x} \comp_{D_x}\right)\ket{+_{\mathcal Y}^{|\mathcal X|}} \\
	                              & =\bigotimes_{x} \comp_{D_x}|+_{\mathcal Y}\rangle=\bigotimes_{x}|\perp\rangle=|\perp^{|\mathcal X|}\rangle,
\end{align*}
which is the quantum representation of the trivial database that maps any $x \in \mathcal{X}$ to $\perp$.

As a consequence, the internal state of the compressed oracle after $q$ queries is supported by computational basis states $|D\rangle$
for which $D(x)=\perp$ (respectively $\hat{D}(x)=\perp$ ) for all but (at most) $q$ choices of $x$. We call the number of registers that are in a state other than $\ket\bot$ the \emph{size} of the database.

An oracle query to the superposition oracle is implemented using the unitary operator  $O$  given by \begin{equation}\label{Oracle Query}
	O:\ket{\mathbf{x}}_X|y\rangle_Y \otimes|H\rangle_D \mapsto\ket{\mathbf{x}}_X|y\oplus H(x)\rangle_Y \otimes|H\rangle_D.
\end{equation}

When the compressed oracle is queried, a unitary $O_{X Y D}$, acting on the query registers $X$ and $Y$ and the oracle register $D$, is applied, given by
$$
	O_{X Y D}=\sum_x|x\rangle\left\langle\left. x\right|_X \otimes O_{Y D_x}^x,\right.
$$
with
$$
	O_{Y D_x}^x=\comp_{D_x} \mathrm{CNOT}_{Y D_x} \comp_{D_x}
$$
where $\operatorname{CNOT}_{Y D_x}|y\rangle\left|y_x\right\rangle=\left|y \oplus y_x\right\rangle\left|y_x\right\rangle$ for $y, y_x \in\{0,1\}^{\ell}$ and acts as identity on $|y\rangle|\perp\rangle$

\textbf{Query algorithms.}
A quantum algorithm interacting with a compressed oracle works as follows. It has registers $XYZ\outputreg$, where $X$ is the query input, $Y$ the query output, $Z$ is the local workspace of the algorithm and \outputreg is the register that is measured to produce the final output. $D$ is the compressed oracle database register that is inaccessible to the algorithm except via oracle queries.

Let $|\phi_i\rangle$ denote the adversary's state before the $(i+1)$-th query, and $|\phi_i'\rangle$ the state immediately after that query, \[
	|\phi_i'\rangle = O_{XYD}|\phi_i\rangle.
\]
The adversary now applies a unitary $U$\footnote{WLOG, this unitary does not depend on $i$: We can include a counter register $Z_N$ in $Z$ which keeps track of the number of queries that have been applied thus far, and set $U$ as a controlled unitary controlled on $Z_N$ and increment $Z_N$.} on $XYZ$ to obtain
\[
	|\phi_{i+1}\rangle = U_{XYZ\outputreg}|\phi_i'\rangle.
\]

Thus, after $q$ queries, the overall computation evolves as:
\begin{enumerate}
	\item Start in some initial state $|\phi_0\rangle$.
	\item For each $1\leq i\leq q$:
	      \begin{enumerate}
		      \item Apply the oracle $O_{XYD}$ to obtain $|\phi_i'\rangle$.
		      \item Apply $U_{XYZ\outputreg}$ to obtain $|\phi_i\rangle$.
	      \end{enumerate}
	\item At the end, perform a computational basis measurement of register $\outputreg$.
\end{enumerate}

\paragraph{Reprogrammming in the QROM.}

We need an adaptive reprogramming lemma from  \cite{AC:GHHM21}.

Let $\mathcal X, \mathcal X', \mathcal Y$ be finite sets. Consider the game
$\texttt{Repro}_b$, where $b \in \{0,1\}$:
\begin{enumerate}
	\item Sample an initial random oracle $O_0 : \mathcal X  \to \mathcal Y$.
	\item Define $O_1$ to be $O_0$, but with access to a
	      \textsf{Reprogram} oracle:
	      on input a probability distribution $p$ on $\mathcal X\times \mathcal X'$, it samples $(x,x')\gets p$ and $y\gets \mathcal Y$, sets $O_1(x) := y$, and returns $(x,x')$.
	\item An adversary $A$ is given oracle access to $O_b$ (quantum) and
	      to \textsf{Reprogram}, and outputs a bit $b'$.
\end{enumerate}
The distinguishing advantage is
\[
	\left| \Pr[\texttt{Repro}^{A}_{1} \implies 1] -
	\Pr[\texttt{Repro}^{A}_{0} \implies 1] \right|.
\]
Let $p_X$ be the marginal of $p$
on $X$. Define
\[
	p^{(r)}_{\max} := \mathbb{E}\left[ \max_{x} p^{(r)}_X(x) \right],
\]
where $p^{(r)}$ is the adversary's $r$-th input to the $\mathsf{Reprogram}$ oracle, and the expectation is over the adversary's randomness up to
round $r$.

\begin{lemma}\label{lem:adarep}[Simplified from Theorem~1 in \cite{AC:GHHM21}]
	For any adversary $D$ making $R$ reprogramming calls and $q$ quantum
	queries in total,
	\[
		\left| \Pr[\texttt{Repro}^{D}_{1} \implies 1] -
		\Pr[\texttt{Repro}^{D}_{0} \implies 1] \right|
		\leq \sum_{r=1}^{R}
		\left( \sqrt{q \cdot p^{(r)}_{\max}}
		+ \tfrac{1}{2} q\cdot p^{(r)}_{\max}
		\right).
	\]
\end{lemma}

\paragraph{Quantum union bound}
We need a lemma referred to as the quantum union bound.
\begin{lemma}[Theorem 1.1 in \cite{o2022quantum}]\label{lem:q-union}
	Let $\ket\psi\in\mathcal H$ be a quantum state and let $P_i$, $i=1,...,N$ be projectors on $\mathcal H$ such that $\bra\psi P_i\ket\psi\ge 1-\varepsilon_i$. Then
	\begin{align*}
		\left\|\left(\prod_{i=1}^NP_i\right)\ket\psi\right\|^2\ge 1-4\sum_{i=1}^N\varepsilon_i.
	\end{align*}
\end{lemma}

\paragraph{Chernoff Bounds.}

The following lemma is taken from combining Theorem 4.4 and 4.5 from \cite{MitzenmacherUpfal05}.
\begin{lemma}[Chernoff Bounds]
	\label{chernoff definition theorem}
	Let $X_1, \ldots, X_n$ be independent Bernoulli variables  such that
	$\Pr(X_i = 1) = p_i$.
	Let
	\[
		X = \sum_{i=1}^n X_i
		\quad \text{and} \quad
		\mu = \mathbb{E}[X].
	\]
	Then the following Chernoff bounds hold:
	\begin{enumerate}

		\item (\textbf{Upper tail}
		      ) For $0 < \delta \le 1$,
		      \[
			      \Pr\!\left(X \ge (1+\delta)\mu\right)
			      \le
			      e^{-\mu \delta^2 / 3}.
			      \tag{4.2}
		      \]

		\item (\textbf{Lower tail}) For $0 < \delta \le 1$,
		      \[
			      \Pr\!\left(X \le (1-\delta)\mu\right)
			      \le
			      e^{-\mu \delta^2 / 2}.
			      \tag{4.4}
		      \]
	\end{enumerate}
\end{lemma}

The random variables $X_i$ defined by sequential measurements form a martingale difference sequence rather than independent trials. We therefore need to apply the following martingale Chernoff bound derived from Azuma-Hoeffding Inequality below, from Theorem 12.6 in \cite{MitzenmacherUpfal05}
\begin{lemma}[Azuma--Hoeffding Inequality]\label{azuma definition}
	Let $(Y_0, Y_1, \ldots, Y_n)$ be a martingale such that for all $i = 1,\ldots,n$,
	\[
		|Y_i - Y_{i-1}| \le c_i \quad \text{almost surely}.
	\]
	Then for any $t > 0$,
	\[
		\Pr\!\left[\, Y_n - Y_0 \ge t \,\right]
		\;\le\;
		\exp\!\left(
		- \frac{t^2}{2 \sum_{i=1}^n c_i^2}
		\right).
	\]

\end{lemma}

\section{Proof of Post Quantum Security of the Fischlin Transform}
We are now set to prove that the Fischlin transform is a NIZK (\cref{Nizk defn}) with straight-line extractability (\cref{def:PoKOnline}) and zero knowledge in Quantum Random Oracle Model.

\subsubsection{Proof Intuition.}
\label{sec:intuition}
We begin by describing the proof strategy for the deterministic prover case
(Section~\ref{subsec: Extactability -- Deterministic Commitments}), which contains the key ideas.

\paragraph{The counting argument and why independence breaks.}
The high-level strategy is a counting argument by contradiction. A valid
Fischlin proof requires $k$ hash outputs to equal $0^\ell$, an event that
occurs independently with probability $2^{-\ell}$ per query. By a standard
Chernoff bound (\cref{chernoff definition theorem}), the total number of $0^\ell$
outcomes in the oracle database concentrates tightly around its mean
$\mu = 2^{-\ell} kN$, so seeing significantly more than $\mu$ registers in
state $|0^\ell\rangle$ is exponentially unlikely (\cref{lem:simple-tail-bound}).
The plan is therefore to show that if the prover succeeds \emph{and} the
extractor fails, the database must contain \emph{many more} than $\mu$
registers in state $|0^\ell\rangle$, contradicting this tail bound.

The difficulty is that conditioning on both events simultaneously destroys
independence. The prover-success projector $P_s$ forces the $k$ oracle
registers corresponding to the verified transcripts into the state $|0^\ell\rangle$.
The extractor-failure projector $E_f$ then further constrains the
\emph{remaining} registers, correlating them with the verified ones. Once
both projectors have been applied, the oracle outputs stored in the database registers are no longer
independent, and the Chernoff bound used in \cref{lem:simple-tail-bound} does not
apply. A new strategy is needed to lower-bound the number of $0^\ell$
outcomes in this correlated state.

\paragraph{Decomposition, symmetrization, and the martingale.}
The resolution proceeds in three steps. First, we \emph{decompose} the
post-success state. When the success projector forces $k$ registers to
$|0^\ell\rangle$ and we switch to the compressed basis, the term
$(\mathrm{Comp}|0^\ell\rangle)^{\otimes k}$ is not simply $|0^\ell\rangle^{\otimes k}$;
it is a superposition that Lemma~\ref{lem:comp 0^k} decomposes by the
\emph{subset $S$ of registers that carry $|0^\ell\rangle$}:
\[
	(\mathrm{Comp}|0^\ell\rangle)^{\otimes k}
	\;=\;
	\sum_{\substack{S \subset [k] \\ |S| \ge (1-\gamma)k}}
	|0^\ell\rangle^{\otimes |S|}_S \,|\Gamma^{(|S|)}\rangle_{S^c}
	\;+\; |\delta_{\gamma,k}\rangle,
\]
where the tail $|\delta_{\gamma,k}\rangle$ is small. This isolates the
``bulk'' of the state, in which at least $(1-\gamma)k$ of the $k$ registers
carry $|0^\ell\rangle$ in the compressed picture, from a negligible
remainder. Crucially, applying $E_f$ to this decomposed state forces most of
the \emph{remaining} $m = k(N-1)$ database registers (those outside the
verified transcript) into the compressed initial state $|\perp\rangle$; and
since $\mathrm{Comp}|\perp\rangle = |{+^\ell}\rangle$, each such register
contributes a $0^\ell$ outcome with probability $2^{-\ell}$ when measured.

Second, we \emph{symmetrize} the resulting state over permutations of these
$m$ registers (Equation~\eqref{eq:sym}). The post-conditioning state is
not independent across registers, so we cannot apply Chernoff directly.
Symmetrization is the substitute: it ensures the state lies in the subspace
$W^m_n$ (where at least $n$ of the $m$ registers are in state
$|{+^\ell}\rangle$), which is closed under permutations. For any such
symmetric state, Lemma~\ref{lem:measure} gives a lower bound on the
probability of obtaining $0^\ell$ when measuring any single register;
and, crucially, the post-measurement state remains in the same class of
subspace $W^{m-1}_{n-1}$, so the bound applies recursively to each
subsequent register.

Third, this recursive structure defines a \emph{martingale}. Let $X_i$
indicate whether measuring the $i$-th register yields $0^\ell$, and let
$Z_i = \mathbb{E}[X_i \mid X_{<i}]$ be the conditional expectation, which
Lemma \ref{lem:measure} lower-bounds at each step. The sum $\hat{Z}_i =
	\sum_{j \le i}(Z_j - X_j)$ is a martingale with bounded differences
$|\hat{Z}_i - \hat{Z}_{i-1}| \le 1$, so the Azuma--Hoeffding inequality
(Lemma~\ref{azuma definition}) applies. Combined with the lower bound on the
cumulative mean $\mu' = \sum_i Z_i$ from Lemma~\ref{lem:cum-mean-bound}, this
gives: conditioned on prover success and extractor failure, the database
contains at least $\eta \approx \mu'$ registers in state $|0^\ell\rangle$
with high probability. Since $\eta > \mu$ under our parameter choices, this
contradicts the upper tail bound of Lemma~\ref{lem:simple-tail-bound}, completing the
argument.

The most involved part of the proof is now done. We lift the result to general malicious provers in \cref{sec:lift-to-general-prover}. This can be done by decomposing a general prover's final state according to when it has first queried an input starting with a fixed commitment vector. As the relevant part of the compressed oracle database is untouched before that query, we can \emph{post-select} on that query with input starting with a fixed commitment vector actually being made, without query cost. This allows the application of the result for deterministic provers to bound the probability of prover success and extractor failure \emph{conditioned} on a fixed commitment vector.

\subsubsection{Prover and Extractor Definitions.}We describe how the prover and extractor work in our analysis for QROM-extractability of the Fischlin Transform described in Definition \ref{Fischlin Protocol}. While it is obviously important for the security of the NIZK that not all transcripts are valid, we will only explicitly keep track of the part of the compressed oracle register that corresponds to transcripts that are accepted by the \sigp-verifier, as this turns out to be sufficient for extractability.

\begin{definition}[Extractor ${\Ext}_{\mathsf{Fis }}$]\label{Fischlin Extractor}
	Let $\Sigma$ be a \sigp. The quantum \\straight-line extractor for the Fischlin transform is given the statement $x$ and simulates the quantum random oracle using a compressed oracle. After the adversary has finished, it receives the proof $\pi$ output by the adversary.

	The extractor works as follows.\\
	$\Ext_{\text {Fis }}(x)$ :\\
	1. The prover receives oracle access to a compressed oracle with database $D$ and produces a proof $\pi$. \\
	2. Verify the proof. If verification fails, abort.\\
	3. Measure  the compressed oracle database register. Let $((\boldsymbol{a}, i, c, z),(\boldsymbol{a}, i, c', z'))$ be the lexicographically first pair of inputs in the measured database such that $(c, z) \neq\left(c', z'\right)$, and \mbox{$\mathcal{V}_{\Sigma}\left(x, a_i, c, z\right)=\mathcal{V}_{\Sigma}\left(x, a_i, c', z'\right)=1$}. If no such query is found, abort.\\
	4. Output $\operatorname{Ext}_{\mathrm{ss}}\left(x,a_i, c, z, c',z'\right)$.
\end{definition}

\subsection{Extractability -- Deterministic Commitments}\label{subsec: Extactability -- Deterministic Commitments}
We first prove that the extractor will succeed in extracting a witness whenever a prover outputs a valid proof \emph{that starts with a fixed vector of commitments}. We will later lift that result to arbitrary provers (see \cref{sec:lift-to-general-prover}).

In this section, we consider therefore a modified NIZK $\Pi_{\mathbf{Fis}}^{(\mathbf a_0)}$ parameterized by a commitment vector $\mathbf a_0$ where the verifier only accepts if the proof is valid \emph{and} the proof starts with the commitment vector $\mathbf a_0$. We can thus further restrict our attention to only the compressed oracle registers for inputs of the form $(\mathbf a, i, c, z)$ such that $\mathbf a=\mathbf a_0$, and $(a_i, c,z)$ is accepted by $\mathcal V_{\Sigma}$. As we assume the \sigp to have unique responses, any input is specified among the remaining set of inputs by the pair $(i,c)\in [k]\times \mathcal C$.

An arbitrary dishonest prover $\prover^*$ for the Fischlin transform is a quantum query algorithm as described in \cref{subsection: QROM}, with the output register $\outputreg$ having the right structure to hold a proof $\pi=(\mathbf a,\mathbf c,\mathbf z)$.
The joint state of the prover and the oracle in the uncompressed basis (i.e., in the superposition oracle picture, not the pre-compressed oracle) after the unitary part of a malicious prover has finished but $\outputreg$ has not been measured to produce the prover's output is denoted as

\begin{equation}\label{|phi>}
	|\phi\rangle=\sum_{\mathbf{x}, H}2^{-\frac{lk\challspacesize}{2}}\ket{\mathbf{x}}_\outputreg|\phi^{\mathbf x,H}\rangle_Z|{H}\rangle_D \end{equation}

Here, $D=(D_{(i,c)})_{i\in[k],c\in\mathcal C}$ and we subsume the malicious prover's query input and output registers as well as the remainder of the compressed oracle database in the register $Z$. The marginal of the register $D$ remains uniform (as the query operator is a controlled unitary with control register $D$).

To bound the probability for the event where the prover succeeds and the extractor fails, the analysis in the following sections is used.

To facilitate our statistical arguments we define the following projector.
$$
	P_{\eta}=\sum_{\bar\omega (H) \geq \eta}\proj{H}_D
$$
where $\bar\omega(H)$ is defined as
$\bar\omega(H)=\left|\left\{\mathbf x \mid H(\mathbf{x})= 0^{\ell}\right\}\right|$.

Applying the projector $\left(P_{\eta}\right)_{D}$ to the state $\ket{\phi}_{\outputreg ZD}$ yields a vector of small norm for $\eta$ significantly bigger than $2^{-l}k\challspacesize$ (where $\challspacesize=\Omega({2^{\ell}\log k})$ is the challenge space size) by a standard Chernoff bound.
\begin{lemma}\label{lem:simple-tail-bound}
	Let $\mu= 2^{-l}k\challspacesize$, $\delta>0$ and $\eta=(1+\delta)\mu$. We have the following probability bound \[
		||P_{\eta}|\phi\rangle||^2\le  \exp\left( -\frac{\delta^2 \mu}{3} \right)\eqqcolon\varepsilon''.
	\]
\end{lemma}
\begin{proof}
	Let $X_{i,c}$, $i\in[k], c\in\mathcal C$ be random variables defined as measurement results of the state $\ket\phi$ by
	\[
		X_{i,c}=\begin{cases}
			1 & \text{ if measuring register $D_{i,c}$ yields outcome $0^{\ell}$} \\
			0 & \text{ else}
		\end{cases},
	\]
	Clearly $||P_{\eta}|\phi\rangle||^2=\sum_{i\in[k], c\in\mathcal C}X_{i,c}\eqqcolon \mathfrak S$. Also, measuring the subregisters of $D$ in the computational basis yields independent uniformly random $l$-bit strings,\footnote{essentially by the correctness of the superposition oracle methodology} so $\Pr[X_{i,c}=1]=2^{-l}$ independently for all $(i,c)$.
	By a standard Chernoff bound, \cref{chernoff definition theorem}, we thus have \js{Notation $S_n$ clash with permutation group, this sum $S_n$ not defined before}\cm{I just used a new font now....}
	\begin{align*}\label{bound e''}
		\mathbb{P}\left[ (\mathfrak S\geq\eta=(1+\delta)\mu) \right]
		\le \mathbb{P}\left[ |\mathfrak S-\mu|\geq\delta\mu \right]\le \exp\left( -\frac{\delta^2 \mu}{3} \right).
	\end{align*}
	\qed
\end{proof}
Thinking classically for a moment, the proof strategy is via contradiction: a malicious prover succeeding \emph{and} the extractor failing with significant probability would imply that $||P_{\eta}|\phi\rangle||^2$ must in fact be large, contradicting the above tail bound. As we are dealing with a malicious quantum prover and with the compressed oracle that is directly manipulated (measured) by the extractor, the quantum analogue of this argument is significantly more involved.

Denote $\ket{\phi_{\mathrm{succ}}}=P_{s}|\phi\rangle$, where
\begin{equation}
	P_s=\sum_{\mathbf{x}} |\mathbf{x}\rangle\langle \mathbf{x}|_{O}\otimes |0^{\ell}\rangle\langle0^{\ell}|^{\otimes k}_{D_{\mathsf{hi}(\mathbf x)}}
\end{equation}
Is the success projector, i.e., the probability that the malicious prover succeeds is $\|\ket{\phi_{\mathrm{succ}}}\|^2$.
Here,
$D_{\mathsf{hi}(\mathbf{x})}$ is the subregister of $D$ containing the hash inputs needed to verify the proof $\mathbf{x}$, i.e.,  $\mathsf{hi}(\mathbf{x})=\left\{(\mathbf a,i,c_i,z_i)\mid i\in[k], V_{\Sigma}(a_i,c_i,z_i)=1\right\}$.

We now see the state of the oracle in the compressed picture. Writing $\mathbf x=(\mathbf{a},\mathbf{c},\mathbf{z})$, we apply $\comp$, \begin{equation}\label{phi'}
	\begin{aligned}
		 & |\phi'_{\mathrm{succ}}\rangle=\comp P_{s}|\phi\rangle =\sum_{\mathbf{x}}|\mathbf{x}\rangle_{O}|\phi^{\mathbf{x}}\rangle_{ZD_{\mathsf{hi}^c}}(\comp|0^{\ell}\rangle^{\otimes k})_{D_{\mathsf{hi}(x)}} \\
		 & =\sum_{\mathbf{x}}|\mathbf{x}\rangle_{O}|\phi^{\mathbf{x}}\rangle_{ZD_{\mathsf{hi}^c}}\left(\sum_{\substack{S \subset\mathsf{hi}(\mathbf x)                                                          \\
				s=|S| \geq (1-\gamma)k}}|0^{\ell}\rangle_{{D_S}}^{\otimes s}|\Gamma\rangle_{D_{\mathsf{hi}(\mathbf x)\setminus S}}+\ket{\delta_{\gamma, k}}_{D_{\mathsf{hi}(\mathbf x)}}\right),
	\end{aligned}
\end{equation}
where we have used \cref{lem:comp 0^k}, and $0<\gamma=\gamma(l,k)\leq\frac{1}{2}$ is chosen later in \cref{lem: negligiblity value lemma}..
\begin{lemma}\label{lem:comp 0^k}There exist sub-normalized states $\ket{\Gamma^{(s)}}$, $0\le s\le k$ such that
	$$(\comp|0^{\ell}\rangle)^{\otimes k}_D=\sum_{\substack{S \subset[k] \\
				s=|S| \geq (1-\gamma)k}}|0^{\ell}\rangle_{{D_S}}^{\otimes s}|\Gamma^{(|S|)}\rangle_{{D_{S^c}}}+\ket{\delta_{\gamma, k}},$$
	such that $\bra{0^\ell}_{D_i}\ket{\Gamma^{(s)}}_{D_{S^c}}=0$ for all $S$ and all $i\in S^c$, where $D=D_1\ldots D_k$,   and the tail bound is given by  $||\ket{\delta_{\gamma, k}}||^2\leq e^{-\left(\gamma-2 \cdot 2^{-l }\right) k / 2 }\eqqcolon \varepsilon_{\gamma,k}$.
\end{lemma}

\begin{proof}
	We bound
	\begin{equation}\label{|0>^k eq}
		\begin{aligned}
			(\comp|0^{\ell}\rangle)^{\otimes k} & =((|\perp\rangle\langle+^{\ell}|+|+^{\ell}\rangle\langle\perp|+(\mathbb{1}-|+^{\ell}\rangle\langle+^{\ell}|-|\perp\rangle\langle\perp|)\ket{0^{\ell}})^{\otimes k} \\
			                                    & =\left(|0^{\ell}\rangle+2^{-l / 2}|\perp\rangle-2^{-l / 2}\ket{+^{\ell}}\right)^{\otimes k}                                                                        \\
			                                    & =\left(|0^{\ell}\rangle+2^{-l / 2}|\perp\rangle-2^{-l }\sum_{z\in\{0,1\}^{\ell}}\ket{z}\right)^{\otimes k}                                                         \\
			                                    & =\left((1-2^{-l})|0^{\ell}\rangle+2^{-l / 2}|\perp\rangle-2^{-l }\sum_{\substack{z\in\{0,1\}^{\ell}                                                                \\z\neq0^{\ell}}}\ket{z}\right)^{\otimes k}\\
			                                    & =\sum_{\substack{S \subset[k]                                                                                                                                      \\
			s=|S| \geq (1-\gamma)k}}|0^{\ell}\rangle_{{D_S}}^{\otimes s}|\Gamma^{(s)}\rangle_{{D_{S^c}}}+\sum_{\substack{S \subset[k]                                                                                \\
			s=|S| < (1-\gamma)k}}|0^{\ell}\rangle_{D_{S}}^{\otimes s}|\Gamma^{(s)}\rangle_{{D_{S^c}}},                                                                                                               \\
			                                    & =\sum_{\substack{S \subset[k]                                                                                                                                      \\
					s=|S| \geq (1-\gamma)k}}|0^{\ell}\rangle_{{D_S}}^{\otimes s}|\Gamma^{(s)}\rangle_{{D_{S^c}}}+\ket{\delta_{\gamma, k}}
		\end{aligned}
	\end{equation}
	where $\bra{0}_{D_i}\ket{\Gamma^{(|S|)}}_{D_{S^C}}=0$  for all $S\subset [k]$ and $i\in S^C$.

	We now bound $||\ket{\delta_{\gamma, k}}||^2$. We can use a standard Chernoff bound for the number of results $0^\ell$ we would obtain were we to measure the entire state in the computational basis. The mean is  \begin{align*}
		\nonumber& \mu=k\left(1-2^{-l }\right)^2 \end{align*}
	Towards using a Chernoff bound we define $\delta$ via $(1-\delta) \mu=(1-\gamma) k$, i.e.,
	\begin{align*}
		\delta & =1-\frac{(1-\gamma) k}{\mu} =1-\frac{(1-\gamma) k}{k p}=1-\frac{(1-\gamma) }{p}
	\end{align*}
	As $||\ket{\delta_{\gamma, k}}||^2$ is the probability of obtaining less then $(1-\gamma)k$ outcomes $0^{\ell}$, we then obtain
	\begin{align*}
		 & ||\ket{\delta_{\gamma, k}}||^2	\le e^{-\delta^2 \mu / 2} =e^{-[1-(1-\gamma) / p]^2 \times k p/2} =e^{-[p-(1-\gamma)]^2 k / 2 p}                                                         \\
		 & \leq e^{-\left[\left(1-2^{-l }\right)^2-(1-\gamma)\right]^2 k / 2 p} \leq e^{-\left[2^{-2l}-2 \cdot 2^{-l }+\gamma\right] k / 2 p} \leq e^{-\left(\gamma-2 \cdot 2^{-l }\right) k / 2 }
	\end{align*}
	\vspace{-1cm}

	\phantom{.}
	\qed
\end{proof}

\subsubsection{Analyzing the extractor.}

Denote by $\mathcal{E}_f$, and $\mathcal{E}_s=\mathbb{1}-\mathcal{E}_f$ the projections onto the subspace where the extractor $\Ext_{\mathsf{Fis}}$ (definition \ref{Fischlin Extractor} ) of the Fischlin Transform fails, and succeeds, respectively. The operator $\mathcal{E}_{f}$ projects onto the states with only one register $D_{\mathbf{a},i,c,z}$ for each $i$ that is not in the state $\ket\perp$ :

\begin{align*}
	(\mathcal{E}_f)_D=\left(\sum_{H: \Ext_{\mathsf{Fis}} \text{ fails}}
	|H\rangle\langle H|\right)_D=\bigotimes_{i=1}^k\left(P_{\bot,1}^{\challspacesize }\right)_{D_{i,*}},
\end{align*}
where $P_{\bot,1}^{\challspacesize }$ is the projector onto the subspace
\begin{align*}
	W_{\bot,1}^{\challspacesize }=\mathrm{span}\{\ket\psi_{D_{i,c}}\ket{\bot^{N-1}}_{D_{i, c^c}}|\ket\psi\in\mathbb C^{2^{\ell}}\},
\end{align*}
and for each repetition index $i\in[k]$, we write
\[
	D_{i,*} \coloneqq \{ D_{(\mathbf a_0,i,c,z)} : c\in\mathcal C \}
\]
for the collection of compressed oracle registers corresponding to all
valid transcripts with fixed commitment vector $\mathbf a_0$ and repetition
index $i$, ranging over all challenges $c\in\mathcal C$. Above, we have also used $D_{i, c^c}=(D_{i,c'})_{c'\in\mathcal C\setminus \{c\}}$.

We need to bound the probability that the extractor fails and the prover succeeds, so we study the norm of the state Here $\comp_*$ indicates compression operator $\comp$ applied to each register of set $*$. \begin{align}\label{eq:EfPs-phi}
	          & \comp_D\nonumber\mathcal{E}_f\comp_D P_s|\phi\rangle                                                                                                                                \\
	\nonumber & =\comp_D \mathcal{E}_f\sum_{\mathbf{x}}|\mathbf{x}\rangle_{\outputreg}|\phi^{\mathbf{x}}\rangle_{ZD_{\mathsf{hi}(x)^C}}\left(\sum_{\substack{S \subset\mathsf{hi}(\mathbf x)        \\
	s=|S| \geq (1-\gamma)k}}\!\!\!\!\!|0^{\ell}\rangle_{{D_S}}^{\otimes s}|\Gamma^{(s)}\rangle_{D_{\mathsf{hi}(\mathbf x)\setminus S}}+\ket{\delta_{\gamma, k}}_{D_{\mathsf{hi}(\mathbf x)}}\right) \\
	          & =\sum_\mathbf{x} |\mathbf{x}\rangle_{\outputreg}\sum_{\substack{S \subset\mathsf{hi}(\mathbf x)                                                                                     \\
			s=|S| \geq (1-\gamma)k}}\left|\psi^{\mathbf{\mathbf{x}}}_{S}\right\rangle_{ZD_{S^C}}\left(\comp|0^{\ell}\rangle\right)_{D_S}^{\otimes s}+\ket{\delta'_{\gamma, k}}_{\outputreg Z D}
\end{align}
where we implicitly define $\left|\psi^{\mathbf{\mathbf{x}}}_{S}\right\rangle_{ZD_{S^C}}=\comp_{D_{S^C}} \left|\hat \psi^{\mathbf{\mathbf{x}}}_{S}\right\rangle_{ZD_{S^C}}$ with \begin{equation*}
	\left|\hat\psi^{\mathbf{\mathbf{x}}}_{S}\right\rangle_{ZD_{S^C}}|0^{\ell}\rangle_{{D_S}}^{\otimes s}=\mathcal{E}_f|\phi^{\mathbf{x}}\rangle_{ZD_{\mathsf{hi}(x)^C}}|\Gamma^{(s)}\rangle_{D_{\mathsf{hi}(x)\setminus S}}|0^{\ell}\rangle_{{D_S}}^{\otimes s}
\end{equation*}
and have set  and
\begin{align}\label{eq:def-delta-prime-gamma-k}
	\ket{\delta'_{\gamma, k}}_{\outputreg Z D}=\comp_D \mathcal{E}_f\sum_{\mathbf{x}}|\mathbf{x}\rangle_{\outputreg}|\phi^{\mathbf{x}}\rangle_{ZD_{\mathsf{hi}(x)^C}}\left(\ket{\delta_{\gamma, k}}_{D_{\mathsf{hi}(\mathbf x)}}\right).
\end{align}
Note that
\begin{align}\label{eq:bound-on-delta-prime-gamma-k}
	\|	\ket{\delta'_{\gamma, k}}\|^2\le \varepsilon_{\gamma, k}
\end{align}
since both $\mathcal E_f$ and $\comp_D$ are contractions,
so the bound follows directly from Lemma~\ref{lem:comp 0^k}.

We will now show that measuring a normalized version of
\begin{equation}\label{eq:def-psi-x}
	\ket{\psi^{\mathbf x}}_{ZD}=\sum_{\substack{S \subset\mathsf{hi}(\mathbf x)\\
			s=|S| \geq (1-\gamma)k}}\left|\psi^{\mathbf{\mathbf{x}}}_{S}\right\rangle_{ZD_{S^C}}\left(\comp|0^{\ell}\rangle\right)_{D_S}^{\otimes s}
\end{equation}
in the computational basis yields many $0$s. Comparing to the final prover-oracle state without conditioning on success and extractor failure yields a bound on $\|\mathcal{E}_f\comp P_s|\phi\rangle\|^2$, i.e., the probability that the prover succeeds, but the extractor still fails.

For $m,n\in\mathbb N$, $m\ge n$, we define the subspaces \begin{align*}
	W_n^m=\mathrm{span}\Big\{ & \ket{+^{\ell}}^{\otimes s}_{D_T}\ket\psi_{D_{T^c}}                                                                                                                  \\
	                          & \quad \Big|T\subset[m], s\coloneqq|T|\ge n, \ket\psi\in\left(\mathbb C^{2^{\ell}}\right)^{\otimes m-s}\Big\}\subseteq\left(\mathbb C^{2^{\ell}}\right)^{\otimes m}.
\end{align*}

The projector $\mathcal E_f$ projects onto a subspace where for each $i$, only one register $D_{\mathbf a, i,c, z}$ is not in the state $\ket\bot$. As $\comp\ket\bot=\ket{+^{\ell}}$ and $1\le i\le k$, we have \begin{equation}\label{psi main}
	\ket{\psi^{\mathbf x}}\in W_n^m \otimes\mathcal H_{D_{\mathsf{hi}(x)}}\otimes \mathcal H_{Z},
\end{equation}
where $m=k(\challspacesize-1)$,  $n=m-\gamma k$, $W_n^m$ is viewed as a subspace of the state space of $D_{\mathsf{hi}(\mathbf x)^c}$, and $\mathcal H_{D_{\mathsf{hi}(x)}}$ and $\mathcal H_{Z}$
are the state spaces of registers $D_{\mathsf{hi}(x)}$ and $Z$, respectively.

In the following, we denote $\tilde D=D_{\mathsf{hi}(x)^{c}}$ and its subregisters by $\tilde D_i$, $1\le i\le m$. For simplicity, we sometimes write $\ket{\alpha}_I\coloneqq\ket\alpha_{\tilde D_I}$ to indicate that for some $I\subset [m]$, the register $\tilde D_I$ is in state $\ket\alpha$.

We want to bound $\left\|\left(P_\eta\right)_{\tilde D}|\psi^{\mathbf x}\rangle\right\|$  for some integer $\eta$. This quantity is invariant under permutation of the $m$ subregisters of $\tilde D$. We therefore analyze a symmetrized version of the state.

Consider the permutation action of any permutation $\pi \in S_m$, where $S_m$ is the group of permutations of $m$ elements,
on the subregisters of $\tilde D$,
$$
	\begin{aligned}
		\pi\left|y_1\right\rangle\left|y_2\right\rangle \ldots\left|y_m\right\rangle =\left|y_{\pi^{-1}(1)}\right\rangle\left|y_{\pi^{-1}(2)}\right\rangle \cdots \left|y_{\pi^{-1}(m)}\right\rangle.
	\end{aligned}
$$
We can now define the state\begin{align}\label{eq:sym}
	\ket{\psi_{\mathrm{sym}}^{\mathbf{x}}}_{ZDG}=\frac{1}{\sqrt{m!}}\sum_{\pi\in S_m}\pi_{\tilde D}\ket{\psi^{\mathbf{x}}}_{ZD}\ket\pi_G,
\end{align}
where $G$ is a register keeping track of which permutation was applied to $\tilde D$,
and observe that $\left\|\left(P_\eta\right)_{\tilde D}|\psi^{\mathbf x}\rangle\right\|=\ifsubmission \else\left\fi\|\left(P_\eta\right)_{\tilde D}\ket{\psi_{\mathrm{sym}}^{\mathbf{x}}}_{ZDG}\ifsubmission \else\right\fi\|$  since $P_\eta$ commutes with permutations of the $m$ subregisters.
We will now lower-bound $\ifsubmission \else\left\fi\|\left(P_\eta\right)_{\tilde D}\ket{\psi_{\mathrm{sym}}^{\mathbf{x}}}_{ZDG}\ifsubmission \else\right\fi\|$ using the following strategy.
\begin{enumerate}
	\item We define a sequence of binary random variables indicating whether measuring the $\tilde D_i$ resulted in $0^{\ell}$ or not
	\item We show a lower bound on the conditional expectations of this sequence of random variables
	\item We use the Azuma-Hoeffding martingale concentration inequality to show a lower bound on the number of $0^{\ell}$ outcomes that holds with high probability
\end{enumerate}
We begin with a lemma that will facilitate the lower bound on the conditional expectation. Informally, this lemma shows that, for a symmetric state $\ket\psi$ supported on subspace $W^m_n$ on register $\tilde D$ which is spanned by states with most sub-registers in state $\ket{+^{\ell}}$, the probability of obtaining $0^{\ell}$ when measuring the first sub-register is not too small. This is because due to the symmetry assumption, the amplitude of the first sub-register being in state $\ket{+^{\ell}}$ is at least $\sqrt{m/n}$.
\begin{lemma}\label{lem:measure}
	Let $\ket\psi_{\tilde DE}\in W_n^m\otimes \mathcal H_E$ be a normalized quantum state on registers $\tilde D=\tilde D_1\ldots \tilde D_m$ and an additional register $E$ such that for all operators $A$ acting on $\tilde D$ and for all $\pi\in S_m$ it holds that $\left\|A_{\tilde D}\pi_{\tilde D}\ket\psi_{\tilde DE}\right\|=\left\|A_{\tilde D}\ket\psi_{\tilde DE}\right\|$. Then
	\begin{align*}
		\left\|\bra{0^{\ell}}_{\tilde D_1}\ket\psi_{\tilde DE}\right\|^2\ge \frac{2^{-l}n}{m}-2^{1-l / 2} \sqrt{\frac{n(m-n)}{m^2}}
	\end{align*}
\end{lemma}
\begin{proof}
	We can decompose
	\begin{align*}
		\ket\psi_{\tilde DE} & \eqqcolon\sum_{\substack{I \subset[m] \\
				|I| \geq n}}\left|+^{\ell}\rangle_{I}| \psi_I\right\rangle_{I^C E}
	\end{align*}
	with the convention that $\langle +^{\ell} |_i |\psi_I\rangle_{I^c E}=0$ for all $i\notin I$. We can use this decomposition to derive expressions for the parts of the state where the first subregister is in state $\ket{+^{\ell}}$, and orthogonal to it, respectively,
	\begin{equation}\label{Psi= psi+ + psi-}
		\begin{aligned}\ket\psi_{\tilde DE} & =\sum_{\substack{I \subset[m]                                                                                               \\
               |I| \geq n}}\ket{+^{\ell}}^{\otimes |I|}_{I}| \psi_I\rangle_{I^CE}=\sum_{\substack{I \subset[m]                                                    \\
               |I| \geq n\wedge 1\in I}}\ket{+^{\ell}}^{\otimes |I|}_{I}| \psi_I\rangle_{I^CE}+\sum_{\substack{I \subset[m]                                       \\
               |I| \geq n\wedge 1\notin I}}\ket{+^{\ell}}^{\otimes |I|}_{I}| \psi_I\rangle_{I^CE}                                                                 \\
                                    & =|+^{\ell}\rangle_{1}\left( \sum_{\substack{I' \subset[m]\backslash{1}                                                      \\
               |I'| \geq n-1}}|+^{\ell}\rangle_{I'}^{\otimes |l'|}| \psi_{I'\cup\left\{1\right\}}\rangle_{{(I'\cup\{1\})^CE}}\right)+\sum_{\substack{I \subset[m] \\
               |I| \geq n\wedge 1\notin I}}\ket{+^{\ell}}^{\otimes |I|}_{I}| \psi_I\rangle_{I^CE}                                                                 \\
                                    & \eqqcolon|+^{\ell}\rangle_{1}|\psi^+\rangle_{1^cE}+|\psi^-\rangle_{\tilde DE},
		\end{aligned}
	\end{equation}
	where $\bra{+^{\ell}}_1|\psi^-\rangle_{\tilde DE}=0$.
	By the assumed permutation symmetry of $\ket\psi$, the norm of $\ket{\psi_I}$ only depends on $|I|$. We thus define $\zeta_i=\|\ket{\psi_{[i]}}\|^2$ and note that by the normalization of $\ket\psi$ we have
	\begin{align*}
		\sum_{i=n}^m\binom{m}{i}\zeta_i=1.
	\end{align*}
	We can now bound
	\begin{align*}
		\|\ket{\psi_+}\|^2= & \sum_{\substack{I' \subset[m]\backslash{1}                                                                           \\
		|I'| \geq n-1}}\|| \psi_{I'\cup\left\{1\right\}}\rangle\|^2=\sum_{i=n}^m\binom{m-1}{i-1}\zeta_i=\sum_{i=n}^m\frac{i}{m}\binom{m}{i}\zeta_i \\
		\ge                 & \sum_{i=n}^m\frac{n}{m}\binom{m}{i}\zeta_i=\frac n m\eqqcolon p
	\end{align*}
	Therefore we get, for some $\alpha\in[0,1]$,
	\begin{align*}
		 & \left\|\bra{0^{\ell}}_{1}\ket{\psi}_{\tilde DE}\right\|^2
		= \|\langle0|+\rangle  \ket{\psi^{+}}_{1^c E}+\bra0_1 \ket{\psi^{-}}_{\tilde DE}\|^2                               \\
		 & \geq\left(2^{-l / 2} \|\left|\psi^{+}\right\rangle\|-\alpha\|\left|\psi^{-}\right\rangle \|\right)^2            \\
		 & \geq\left(2^{-l / 2} \sqrt{p}-\alpha\sqrt{1-p}\right)^2 =2^{-l} p+\alpha^2(1-p)-\alpha2^{1-l / 2} \sqrt{p(1-p)} \\
		 & \geq \frac{2^{-l}n}{m}-2^{1-l / 2} \sqrt{\frac{n(m-n)}{m^2}}.
	\end{align*}
	\qed\end{proof}

\subsubsection{Bounding $\left\|\left(P_\eta\right)_{\tilde D}|\psi^{\mathbf x}\rangle\right\|$ using martingale tail bound.}
We will now follow the same strategy as in \cref{lem:simple-tail-bound} to show that after applying the extractor failure projector, there are, in fact, more than $\eta$ registers in state $\ket{0^{\ell}}$ with high probability. There is, however, one crucial difference: the random variables obtained by measuring the sub-registers of $\tilde D$ in the state $\ket{\psi^{\mathbf x}}$ are not independent.
This renders the standard Chernoff bound we used in \cref{lem:simple-tail-bound} inapplicable. We therefore resort to a martingale tail bound for our analysis.

In the following, for ease of notation, let
\[
	\ket{\psi}=	\frac{1}{\|\ket{\psi_{\mathrm{sym}}^{\mathbf{x}}}\|}\ket{\psi_{\mathrm{sym}}^{\mathbf{x}}}.
\]

Let $X_i$, $0 \leq i \leq k(\challspacesize-1)$ be the random variables\footnote{Following established conventions, we denote both the query input register and certain random variables using the letter $X$. Abuse of notation is avoided as the random variables have subscripts and context should help the reader avoid confusion.} defined analogously to the proof of \cref{lem:simple-tail-bound}, i.e.,
\begin{equation*}
	X_i=\begin{cases}
		1 & \text{ if measuring register $\tilde D_i$ of state $\ket{\psi}$ yields $0^{\ell}$} \\
		0 & \text{else}.
	\end{cases}
\end{equation*}
We define $X_{<i}=(X_1,\ldots X_{i-1})$, and
\begin{align}\label{Xi,yi}
	Z_i       & =\mathbb{E}\left[X_i| X_{<i}\right], \\
	\hat{X_i} & =Z_i-X_i , \text{ and }              \\
	\hat{Z_i} & =\sum_{j=1}^i\hat{X_j}.
\end{align}
Note that conditioning on $X_{<i}$, $Z_{<i}$, $\hat X_i$ and $\hat Z_{<i}$ are all equivalent. Let the value of the standard cumulative mean be denoted as
\begin{equation}\label{mu'}
	\mu' =\sum_{i=1}^n\mathbb{E}[X_i \mid X_{<i}]=\sum_{i=1}^n Z_i.
\end{equation}

\begin{lemma}\label{lem:martingale-tail-bound}
	We have the following probability bound \[
		\left\|\left(P_\xi\right)_{\tilde D}|\psi\rangle\right\|^2\ge 1-\exp\left(-\frac{{\delta'}^2\underline\mu'^2}{2 m}\right)\eqqcolon 1-\varepsilon'
	\]
	Here,
	\[\underline \mu'\coloneqq 2^{-l}k\left[\challspacesize-1-\gamma\left(1+4\cdot 2^{\ell}+\log\left(\frac{4\gamma}{\log{k}}\right)\right)-4\sqrt{2^{\ell}\gamma\challspacesize}\right]\]
	and
	\[
		\xi=(1-\delta')\underline \mu'.
	\]
\end{lemma}
\begin{proof}
	As in the proof of \cref{lem:simple-tail-bound}, we use the fact that \js{added third term}
	\[
		\left\|\left(P_\xi\right)_{\tilde D}|\psi \rangle\right\|^2=\Pr\left[\sum_{i=1}^m X_i\ge \xi\right]=1-\Pr\left[\sum_{i=1}^m X_i< \xi\right]
	\]
	Note that conditioning on $X_{<i}$, $Z_{<i}$, $\hat X_i$ and $\hat Z_{<i}$ are all equivalent.
	We get
	\begin{align*}
		\mathbb{E}\!\left[ \hat X_i \mid \hat Z_{< i} \right]
		 & =\mathbb{E}\!\left[Z_i-X_i\mid  X_{< i} \right]              \\
		 & = \mathbb{E}\!\left[
		\mathbb{E}\left[X_i| X_{<i}\right]-X_i\middle|\, X_{< i}\right] \\&=\mathbb{E}\left[X_i| X_{<i}\right]-\mathbb{E}\left[X_i| X_{<i}\right]\\
		 & = 0,                                                         \\
		\implies\mathbb{E}\!\left[ \hat Z_i \mid \hat Z_{< i} \right]
		 & = \hat Z_{i-1}
		+ \mathbb{E}\!\left[ \hat X_i \mid \hat Z_{< i} \right]
		= \hat Z_{i-1}
	\end{align*}
	Hence $\hat Z_{i}$ is a martingale. Additionally, $\hat Z_0=0$ and $|\hat Z_i-\hat Z_{i-1}|=|\hat X_i|\le 1$. Applying \cref{azuma definition} we thus get
	\begin{align*}
		\Pr\!\left[\, \sum_{j=1}^m X_j \le \mu'-\varepsilon \,\right] & =	\Pr\!\left[\, \sum_{j=1}^m X_j \le \sum_{j=1}^m Z_j-\varepsilon \,\right] \\
		                                                              & =\Pr\!\left[\, 0-\hat Z_m \le -\varepsilon \,\right]                        \\
		                                                              & \le\;
		\exp\!\left(
		- \frac{\varepsilon^2}{2m}
		\right).
	\end{align*}
	Using the lower bound on the cummulative mean from \cref{lem:cum-mean-bound} we obtain the desired bound,
	\begin{align*}
		\Pr\!\left[\, \sum_{j=1}^m X_j \le \underline{\mu}'-\varepsilon \,\right] & \le\;
		\exp\!\left(
		- \frac{\varepsilon^2}{2m}
		\right).
	\end{align*}
	\vspace{-1cm}

	\phantom{.}
	\qed\end{proof}

We continue by providing a lower bound on the cumulative mean.

\begin{lemma}\label{lem:cum-mean-bound}
	The cummulative mean defined in \cref{mu'} is bounded as
	\begin{align*}
		\nonumber\mu' & \geq \underline \mu'=2^{-l}k\left[\challspacesize-1-\gamma\left(1+4\cdot 2^{\ell}+\log\left(\frac{4\gamma}{\log{k}}\right)\right)-4\sqrt{2^{\ell}\gamma\challspacesize}\right].
	\end{align*}
\end{lemma}
The proof proceeds by exploiting \cref{lem:measure} and elementary bounds, and can be found in \cref{app:mu'bound} in the Appendix.

\subsubsection{QROM Extractability Bound} We are now finally ready to prove our main QROM extractability  result for deterministic commitments. For brevity, define
$\CompOperatorComp{E}_f=\comp\mathcal{E}_f\comp$ and $\CompOperatorComp{E}_s=\comp\mathcal{E}_s\comp$. We make the convention that any projectors act on register $D$ unless specified otherwise.

Recall that the argument is a sandwich between two tail bounds pulling in opposite
directions. For a fixed commitment vector $a_0$, we want to bound $\|\tilde{E}_f|\phi_{succ}\rangle\|^2$, the probability
that the prover succeeds but the extractor fails. (Recall the (subnormalized) state  $|\phi_{succ}\rangle =
	P_s|\phi\rangle$ projected onto prover
success. The key quantity is a
threshold $\eta$ on the number of database registers in state $|0^\ell\rangle$.
On one hand, Lemma~\ref{lem:simple-tail-bound} gives an \emph{upper} bound: unconditionally,
the probability of seeing $\geq \eta$ such registers is exponentially small
(call it $\varepsilon''$), because the database registers start independent and
$\eta$ is chosen well above the mean $\mu = 2^{-\ell}kN$. On the other hand,
Lemma~\ref{lem:ineq 33} - which collects the decomposition, symmetrization,
and martingale argument of Section~\ref{sec:intuition} - gives a matching
\emph{lower} bound: conditioned on both prover success and extractor failure,
seeing $\geq \eta$ registers in state $|0^\ell\rangle$ is likely, with
probability at least $(1 - \varepsilon')\|\tilde{E}_f|\phi_{succ}\rangle\|^2$,
up to small error terms from the compressed-basis tail $\varepsilon_{\gamma,k}$.
Combining the two via the triangle inequality, together with a bound on the
cross term $\langle\phi_{succ}|(\mathds 1-\tilde{E}_f )P_\eta \tilde{E}_f|\phi_{succ}\rangle$
from Lemma~\ref{lem:ineq 33}, yields
$\varepsilon'' \geq (1 - \varepsilon')\|\tilde{E}_f|\phi_{succ}\rangle\|^2$,
from which
\[
	\|\tilde{E}_f|\phi_{succ}\rangle\|^2 \;\leq\; \frac{\varepsilon}{1 - \varepsilon}
\]
follows immediately, where $\varepsilon=\max(\varepsilon', \varepsilon'')$. Note that both $\varepsilon''$ and $\varepsilon'$ are negligible in $k$, making the
extraction error negligible. \begin{lemma}[Extractability -- Deterministic Commitments]\label{Extractability main bound}
	Let $\mathcal P $ be a dishonest prover for the Fischlin transform of a special-sound \sigp for a relation $R$ with unique responses. For some constant $c$, let the challenge space be of size  $\challspacesize=c\cdot 2^{\ell}\cdot \log{(k)}$. Let further $2^{\frac 1 c}\le k\le 2^{\frac{2^{\ell}}{256\cdot c}}$ and $l\ge 14$.
	The probability of the event where the prover succeeds and the extractor fails is bounded by
	$$\mathbb{P}[\text{Extractor fails and Prover Succeeds}]=||\CompOperatorComp{E}_fP_{s}|\phi\rangle||^2\leq{\varepsilon/(1-\varepsilon)=\operatorname{negl}(k)},$$
	where $\varepsilon\leq 3 \exp\left[-\frac{ k  }{128 \cdot c\cdot 2^{\ell}\cdot \log k}\right]+7\exp\left[-\frac{ k}{8\cdot 2^{\ell}} \right]$. \end{lemma}
Note that the parameter restrictions allow picking $l=\Theta(\log\log k)$ which indeed makes the error negligible in $k$.
\begin{proof}
	Using \cref{lem:simple-tail-bound}, we bound
	$$
		\begin{aligned}
			\varepsilon'' & \geq ||P_{\eta}\ket{\phi}||^2                                                                                                                                                                                                                                                               \\
			              & \geq ||P_sP_{\eta}|\phi\rangle||^2= || P_{\eta}P_s|\phi\rangle||^2                                                                                                                                                                                                                          \\
			              & = ||P_{\eta}({\CompOperatorComp{E}}_s+{\CompOperatorComp{E}}_f)|\phi_{\mathrm{succ}}\rangle||^2                                                                                                                                                                                             \\
			              & \geq ||P_{\eta}\CompOperatorComp{E}_s|\phi_{\mathrm{succ}}\rangle||^2+||P_{\eta}\CompOperatorComp{E}_f|\phi_{\mathrm{succ}}\rangle||^2-2\operatorname{Re}\left\langle\phi_{\mathrm{succ}}\right|\CompOperatorComp{E}_sP_{\eta}\CompOperatorComp{E}_f\left|\phi_{\mathrm{succ}}\right\rangle \\
			              & \geq ||P_{\eta}\CompOperatorComp{E}_f|\phi_{\mathrm{succ}}\rangle||^2-2\operatorname{Re}\left\langle\phi_{\mathrm{succ}}\right|\CompOperatorComp{E}_sP_{\eta}\CompOperatorComp{E}_f\left|\phi_{\mathrm{succ}}\right\rangle                                                                  \\
			              & \geq (1-\varepsilon')\|\CompOperatorComp{E}_f|\phi_{\mathrm{succ}}\rangle||^2-4\sqrt{\varepsilon_{\gamma,(1-\gamma)k}}-2\operatorname{Re}\left\langle\phi_{\mathrm{succ}}\right|\CompOperatorComp{E}_sP_{\eta}\CompOperatorComp{E}_f\left|\phi_{\mathrm{succ}}\right\rangle                 \\
			              & \text{ (here $\exp \left[-\frac{\delta'^2\mu'}{2m}\right]=\varepsilon'$}\text{as in lemma \ref{lem:ineq 33}, applying lemma \ref{lem:ineq 33})}                                                                                                                                             \\
			              & \geq (1-\varepsilon')\|\CompOperatorComp{E}_f|\phi_{\mathrm{succ}}\rangle||^2-4\sqrt{\varepsilon_{\gamma,(1-\gamma)k}}-\sqrt{2{\varepsilon'}+8\sqrt{\varepsilon_{\gamma,(1-\gamma)k}}}\text{ (due to (\ref{re<-> bound}) below)}\end{aligned}
	$$
	In the second line, we have used that $P_\eta$ and $P_s$ are both diagonal in the standard basis on $D$.
	The bound on the inner product expression is obtained as
	\begin{align}
		 & \operatorname{Re}\nonumber\langle\phi_{\mathrm{succ}}|\CompOperatorComp{E}_s P_{\eta}\CompOperatorComp{E}_f|\phi_{\mathrm{succ}}\rangle \leq \left\langle\phi_{\mathrm{succ}}\right|(\mathbb{I}-\CompOperatorComp{E}_f )\CompOperatorComp{E}_f\left|\phi_{\mathrm{succ}}\right\rangle+ \left\langle\phi_{\mathrm{succ}}\right|(\mathbb{I}-\CompOperatorComp{E}_f) (P_{\eta}\CompOperatorComp{E}_f-\CompOperatorComp{E}_f)\left|\phi_{\mathrm{succ}}\right\rangle \\
		 & \label{re<-> bound}\leq ||(\mathbb{I}-\CompOperatorComp{E}_f)||||(P_{\eta}\CompOperatorComp{E}_f-\CompOperatorComp{E}_f)|\phi_{\mathrm{succ}}\rangle||\leq \sqrt{2{\varepsilon'}+8\sqrt{\varepsilon_{\gamma,(1-\gamma)k}}},\end{align}
	where the last inequality is a gentle-measurement-type bound which we formulate as a specialized lemma, \cref{P_n ext norm lemma}, in \cref{sec:gentle} in the Appendix for convenience. $$\begin{aligned}
			\implies\mathbb{P} & [\text{Extractor fails and Prover Succeeds}]=||\CompOperatorComp{E}_fP_s|\phi\rangle||^2                                                                                                 \\
			                   & \leq \frac{\varepsilon''+\sqrt{2{\varepsilon'}+8\sqrt{\varepsilon_{\gamma,(1-\gamma)k}}}+4\sqrt{\varepsilon_{\gamma,(1-\gamma)k}}}{1-\varepsilon'}\leq \frac{\varepsilon}{1-\varepsilon}
		\end{aligned}
	$$
	$$\begin{aligned}
			\text{  where } \varepsilon & =\operatorname{max}\{\varepsilon', \varepsilon''+\sqrt{2{\varepsilon'}+8\sqrt{\varepsilon_{\gamma,(1-\gamma)k}}}+4\sqrt{\varepsilon_{\gamma,(1-\gamma)k}}\}                   \\
			                            & =\varepsilon''+\sqrt{2{\varepsilon'}+8\sqrt{\varepsilon_{\gamma,(1-\gamma)k}}}+4\sqrt{\varepsilon_{\gamma,(1-\gamma)k}}                                                       \\
			                            & \leq\varepsilon''+\sqrt{2{\varepsilon'}}+\sqrt{8\sqrt{\varepsilon_{\gamma,(1-\gamma)k}}}+4\sqrt{\varepsilon_{\gamma,(1-\gamma)k}} \text{ ( as $\sqrt{a^2+b^2}\leq |a|+|b|$ )} \\
			                            & \leq \varepsilon''+2\sqrt{{\varepsilon'}}+7\sqrt{\sqrt{\varepsilon_{\gamma,(1-\gamma)k}}}                                                                                     \\
			                            & \leq 3 \exp\left[-\frac{ k  }{128 \cdot c\cdot 2^{\ell}\cdot \log k}\right]+7\exp\left[-\frac{ k}{8\cdot 2^{\ell}} \right].
		\end{aligned}$$
	The last inequality is obtained by crudely balancing errors, which is done in  \cref{lem: negligiblity value lemma} in the Appendix.
	\qed
\end{proof}

\begin{lemma}[Using the martingale bounds]\label{lem:ineq 33}$$||P_{\eta}\CompOperatorComp{E}_f|\phi_{\mathrm{succ}}\rangle||^2\geq
		(1-\varepsilon')||\CompOperatorComp{E}_f\ket{\phi_{\mathrm{succ}}} \|^2-4\sqrt{\varepsilon_{\gamma,(1-\gamma)k}}$$
	where $\varepsilon_{\gamma,(1-\gamma)k}$ is tail bound as in lemma \ref{lem:comp 0^k}, $\underline \mu'$ is as in  \cref{lem:martingale-tail-bound}, \js{added $\leq \eta$, edited a bit here}\cm{I reverted your edit, the whole point is that we get $(1-2\gamma)k$ zeroes for free here and only the remaining $\eta-(1-2\gamma)$ ones are from the martingale tail bound}
	\[
		\eta-(1-2\gamma)k =(1-\delta')\underline \mu'
		\quad \text{and}\quad
		\varepsilon' =\exp\left(-\frac{{\delta'}^2\underline\mu'^2}{2 m}\right)
	\]
\end{lemma}
\begin{proof}
	Consider $\CompOperatorComp{E}_fP_s\ket{\phi}=\CompOperatorComp{E}_f\ket{\phi_{\mathrm{succ}}}$  and
	\begin{align*}
		\ket{\tilde{\phi}^{\prime}} & =\frac{\CompOperatorComp{E}_f\left|\phi_{\mathrm{succ}}\right\rangle}{\| \CompOperatorComp{E}_f|\phi_{\mathrm{succ}}\rangle \|}                                                                                       \\
		                            & = \sqrt{\Gamma}\sum_{\mathbf x}\sqrt{p_{\mathbf x}}\ket{\mathbf x}_{\outputreg}\ket{\psi_{\mathrm{norm}}^{\mathbf x}}_{ED}+\frac{\ket{\delta'_{\gamma,k}}}{\| \CompOperatorComp{E}_f|\phi_{\mathrm{succ}}\rangle \|}.
	\end{align*}
	Here we have used \cref{eq:EfPs-phi,eq:def-delta-prime-gamma-k,eq:def-psi-x} and defined
	\begin{align*}
		\ket{\psi_{\mathrm{norm}}^{\mathbf x}} & =\frac{\ket{\psi^{\mathbf x}}}{\|\ket{\psi^{\mathbf x}}\|},\text{ the probability distribution} \\
		p_{\mathbf x}                          & =\frac{\|\ket{\psi^{\mathbf x}}\|^2}{\sum{\mathbf x}\|\ket{\psi^{\mathbf x}}\|^2} \text{ and }
		\Gamma=\frac{\sum{\mathbf x}\|\ket{\psi^{\mathbf x}}\|^2}{\| \CompOperatorComp{E}_f|\phi_{\mathrm{succ}}\rangle \|^2}.
	\end{align*}

	We bound \begin{align*}
		\|P_\eta	\ket{\tilde{\phi}^{\prime}}\|^2 & \ge \Gamma\mathbb E_{\mathbf x\leftarrow p}\left[\|P_\eta\ket{\psi_{\mathrm{norm}}^{\mathbf x}}\|^2\right]-2\sqrt{\varepsilon_{\gamma,k}}                                                              \\
		                                         & =	\Gamma\mathbb E_{\mathbf x\leftarrow p}\left[\left\|P_\eta\frac{1}{\|\ket{\psi_{\mathrm{sym}}^{\mathbf{x}}}\|}\ket{\psi_{\mathrm{sym}}^{\mathbf{x}}}\right\|^2\right]-2\sqrt{\varepsilon_{\gamma,k}}
	\end{align*}
	where we have used \cref{eq:bound-on-delta-prime-gamma-k}. Here, $\mathbf x\gets p$ denotes sampling $\mathbf x$ from the probability distribution $p$. Again using the simplifying notation\\ $\ket\psi=\frac{1}{\|\ket{\psi_{\mathrm{sym}}^{\mathbf{x}}}\|}\ket{\psi_{\mathrm{sym}}^{\mathbf{x}}}$ we bound
	\begin{align*}
		\|P_\eta\ket\psi_{ED}\|^2 & \ge \|\left(P_{(1-2\gamma)k}\right)_{D_{\mathsf{hi}(\mathbf x)}}\left(P_{\eta-(1-2\gamma)k}\right)_{D_{\mathsf{hi}(\mathbf x)^c}}\ket\psi_{ED}\|^2
	\end{align*}
	as choosing particular registers to be in $\ket{0^{\ell}}$ makes the norm smaller. Defining
	\begin{align*}
		\ket{\psi'}=\left(P_{\eta-(1-2\gamma)k}\right)_{D_{\mathsf{hi}(\mathbf x)^c}}\ket\psi_{ED}
	\end{align*}
	we have that \js{edited line below}\cm{reverted}
	\begin{align*}
		\|	\ket{\psi'}\|^2\ge 1-\Pr\left[\sum_{i=1}^m X_i\le \eta-(1-2\gamma)k\right]\geq 1-\varepsilon'
	\end{align*}
	by applying \cref{lem:martingale-tail-bound} with $\xi=\eta-(1-2\gamma)k$, and
	\begin{align*}
		\ket{\psi'}=\sum_{\substack{S \subset\mathsf{hi}(\mathbf x) \\
				s=|S| \geq (1-\gamma)k}}\left|\tilde \psi^{\mathbf{\mathbf{x}}}_{S}\right\rangle_{ZD_{S^C}}\left(\comp|0^{\ell}\rangle\right)_{D_S}^{\otimes s}
	\end{align*}
	for some states $\left|\tilde \psi^{\mathbf{\mathbf{x}}}_{S}\right\rangle_{ZD_{S^C}}$ due to \cref{eq:EfPs-phi}.
	Applying \cref{lem:comp 0^k} with $(1-\gamma)k$ in place of $k$ we get
	\begin{align}\label{eq:applying-comp-0-in-the-end}
		\ket{\psi'}=\sum_{\substack{S \subset\mathsf{hi}(\mathbf x) \\
				s=|S| \geq (1-2\gamma)k}}\left|\tilde{\tilde \psi}^{\mathbf{\mathbf{x}}}_{S}\right\rangle_{ZD_{S^C}}|0^{\ell}\rangle_{D_S}^{\otimes s}+\ket{\delta'_{\gamma,(1-\gamma)k}}
	\end{align}
	for some states $\tilde{\tilde \psi}^{\mathbf{\mathbf{x}}}$, where
	\begin{align*}
		\|\ket{\delta_{\gamma, (1-\gamma)k}}\|^2\le \varepsilon_{\gamma,(1-\gamma)k}.
	\end{align*}
	It follows that
	\begin{align*}
		\|\left(P_{(1-2\gamma)k}\right)_{D_{\mathsf{hi}(\mathbf x)}}\left(P_{\eta-(1-2\gamma)k}\right)_{D_{\mathsf{hi}(\mathbf x)^c}}\ket\psi_{ED}\|^2 & = \|\left(P_{(1-2\gamma)k}\right)_{D_{\mathsf{hi}(\mathbf x)}}\ket{\psi'}\|^2 \\
		                                                                                                                                               & \ge 1-\varepsilon'-2\varepsilon_{\gamma,(1-\gamma) k}.
	\end{align*}
	In summary, we have
	\begin{align*}
		\|P_\eta\ket\psi_{ED}\|^2 & \ge 1-\varepsilon'-2\varepsilon_{\gamma,(1-\gamma) k}
	\end{align*}
	And thus
	\begin{align*}
		||P_{\eta}\CompOperatorComp{E}_f|\phi_{\mathrm{succ}}\rangle||^2 & \ge \|\CompOperatorComp{E}_f\ket{\phi_{\mathrm{succ}}} \|^2(1-\varepsilon'-2\varepsilon_{\gamma,(1-\gamma) k})-2\sqrt{\varepsilon_{\gamma,k}} \\
		                                                                 & \ge \|\CompOperatorComp{E}_f\ket{\phi_{\mathrm{succ}}} \|^2(1-\varepsilon')-4\sqrt{\varepsilon_{\gamma,(1-\gamma )k}}.
	\end{align*}
	\qed
\end{proof}

\subsection{Extractability -- lifting to general prover.}\label{sec:lift-to-general-prover}
We now go on to apply our bound for a prover outputting a proof with a fixed value for $\mathbf a$ to obtain a bound for a general prover. To this end, we prove a bound on the probability that a general prover succeeds but the extractor fails, in terms of the probability that the same happens for a fixed vector of commitments $\mathbf a$.
\begin{theorem}\label{thm:genprover}
	Let $\mathcal P $ be a dishonest prover for the Fischlin transform of a special-sound \sigp for a relation $R$ with unique responses, in the QROM, making no more than $q$ quantum  queries to the random oracle. We have
	\begin{equation}\label{eq:extred-1}
		\Pr_{(x,\pi,w)\leftarrow\langle \mathcal P, \mathcal E\rangle}[\mathcal V(x,\pi)=1]=\Pr_{(x,\pi)\leftarrow\mathcal P^H}[\mathcal V(x,\pi)=1]
	\end{equation}
	and
	\begin{equation}\label{eq:extred-2}
		\Pr_{(x,\pi,w)\leftarrow\langle \mathcal P, \mathcal E\rangle}[\mathcal V(x,\pi)=1\wedge (x,w)\notin R]\le 4(q+k)^2\cdot \delta,
	\end{equation}
	where $\delta$ is such that for any prover $\mathcal P'$ and any $\mathbf a_0\in\mathcal M^{k}$,
	\begin{equation}\label{eq:extred-assump1}
		\Pr_{(x,\pi=(\mathbf a,\mathbf c,\mathbf z),w)\leftarrow\langle \mathcal P', \mathcal E\rangle}[\mathbf a=\mathbf a_0\wedge \mathcal V(x,\pi)=1\wedge (x,w)\notin R]\le\delta.
	\end{equation}
\end{theorem}
Combining with \cref{Extractability main bound}, we get our main result.
\begin{corollary}[straight-line extractability of the Fischlin transform in the QROM]\label{cor:main}
	Let $\Sigma$ be a \sigp for a relation $R$ with special soundness and unique responses. If the challenge space is of size  $\challspacesize=c\cdot 2^{\ell}\cdot \log{(k)}$ for some constant $c$ and parameters $k,l$ of the Fischlin transform,
	$2^{\frac 1 c}\le k\le 2^{\frac{2^{\ell}}{256\cdot c}}$, $l\ge 14$, and $l=O(\log\log k)$ then $\mathsf{Fis}[\Sigma]$ is a proof of knowledge with straight-line extractability. The extractor's simulation of the quantum-accesible random oracle can be made perfect ($\varepsilon_{\mathrm{sim}}=0$) and for a $q$-query adversarial prover,
	\begin{align*}
		\varepsilon_{\mathrm{ex}}\le q^2\cdot \negl(k)
	\end{align*}
\end{corollary}
A quick note on instantiation. Given a \sigp $\hat \Sigma$  with constant challenge space of size $\hat\challspacesize$ we can pick $k$ (not too small) first.  Now we can choose $c$ and $l=\lceil\log\log k+\log c+8\rceil$, then set $r=\lceil\log_{\hat \challspacesize}(c\cdot 2^{\ell}\cdot \log(k))\rceil$ to allow for the desired challenge space size, and define the \sigp $\Sigma$ as the $r$-fold parallel repetition of $\hat \Sigma$ with the challenge space artificially restricted to size $c\cdot 2^{\ell}\cdot \log(k)$.\footnote{We ignore rounding here, which changes $c$ slightly.} The combination of $\Sigma$, $k,l,c$ now fulfils the prerequisites of \cref{cor:main}. As $k$ is a lower bound for the security parameter, this calculation shows that our result supports a parameter region that allows asymptotic security (and thus arbitarily high concrete levels of security).
\begin{proof}[of \cref{thm:genprover}]
	The statement in Equation \eqref{eq:extred-1} is immediate by the perfect correctness of the compressed oracle.

	We begin by constructing a family of provers that either output a valid proof with a fixed commitment vector $\hat{\mathbf a}$, or a dummy symbol. Construct prover $\tilde{\mathcal P}_{\hat{\mathbf a}}$ as follows.
	\begin{enumerate}
		\item Run $(x,\pi=(\mathbf a,\mathbf c,\mathbf z))\leftarrow\mathcal P^H$.
		\item Run $b\leftarrow \mathcal V^H(x,\pi)$.
		\item If $b=1$ and $\mathbf a=\hat{\mathbf a}$, output $(x,\pi)$, else, output $\bot$.
	\end{enumerate}
	Let $\mathcal P_{\hat{\mathbf a}}$ be a version of $\tilde{\mathcal P}_{\hat{\mathbf a}}$ with all measurements delayed to the end, i.e., where verification is run coherently. Observe that $\mathcal P_{\hat{\mathbf a}}$ makes at most $q'=q+k$ queries. Without loss of generality, $\mathcal P_{\hat{\mathbf a}}$ starts with some initial state $\ket{\phi_0}$, applies $\left(UO^H\right)^{q'}$ and measures the output register $\outputreg$. We now consider $\mathcal P_{\hat{\mathbf a}}$ as run by the extractor, i.e., it gets query access to the compressed oracle instead of a standard quantum-accessible random oracle. Denote by $\Xi=|\mathcal M|^k\cdot k\cdot \challspacesize \cdot |\mathcal Z|$ the cardinality of the  domain of the random oracle. We decompose the final state after the unitary part of $\mathcal P_{\hat{\mathbf a}}$ as
	\begin{align*}
		  & \left(UO^H\right)^{q'}\ket{\phi_0}_{XYE}\ket{+^{\ell}}^{\otimes \Xi}_{D}                                                                                                                                            \\
		= & \left(UO^H(\mathds 1-\proj{\hat{\mathbf a}})_{X_{\mathcal M}}\right)^{q'}\ket{\phi_0}_{XYE}\ket{+^{\ell}}^{\otimes \Xi}_{D}                                                                                         \\
		  & \quad+\sum_{i=1}^{q'}\left(UO^H\right)^{q'-i+1}\proj{\hat{\mathbf a}}_{X_{\mathcal M}}\left(UO^H(\mathds 1-\proj{\hat{\mathbf a}}\right)_{X_{\mathcal M}})^{i-1}\ket{\phi_0}_{XYE}\ket{+^{\ell}}^{\otimes \Xi}_{D}.
	\end{align*}
	Here, $E$ is the adversary's working register. Note that by construction, $\mathcal P_{\hat{\mathbf a}}$ always queries an input starting with the $\mathbf a$ in its output proof (if it outputs one). Therefore we have
	\begin{align*}
		P_s\left(UO^H(\mathds 1-\proj{\hat{\mathbf a}}\right)_{X_{\mathcal M}})^{q'}\ket{\phi_0}_{XYE}\ket{+^{\ell}}^{\otimes \Xi}_{D}=0.
	\end{align*}
	Now note that by construction,
	\begin{align*}
		\proj{\hat{\mathbf a}}_{X_{\mathcal M}}\left(UO^H(\mathds 1-\proj{\hat{\mathbf a}}\right)_{X_{\mathcal M}})^{i-1}\ket{\phi_0}_{XYE}\ket{+^{\ell}}^{\otimes \Xi}_{D}=\ket{\tilde\psi_i}_{XYED_{\hat{\mathbf a}}^c}\ket{+^{\ell}}^{\otimes ([k]\cdot \challspacesize \cdot |\mathcal Z|)}_{D_{\hat{\mathbf a}}},
	\end{align*}
	for some sub-normalized state $\ket{\tilde\psi_i}$, as before each query, the input register is projected away from $\ket{\hat{\mathbf a}}$.\footnote{Sub-normalization follows from $\sum_i |\ket{\tilde\psi_i}|^2 \le 1$ which holds by orthogonality of the decomposition.} Let $\ket{\psi_i}=\|\ket{\tilde \psi_i}\|^{-1}\ket{\tilde \psi_i}$. We now define, for each pair $(\hat{\mathbf a}, i)\in\mathcal M^k\times [q]$, a new prover $\mathcal P_{\hat{\mathbf a}, i}^H$ as follows.
	\begin{enumerate}
		\item Prepare $\ket{\psi_i}$ as follows:
		      \begin{enumerate}
			      \item Initialize a compressed oracle for a random function $H'$ identically distributed as $H$.
			      \item Run $\mathcal P^{H'}_{\hat{\mathbf a}}$ until right before the $i$-th query, but for each query, apply the binary measurement whether the query input starts with $\hat{\mathbf a}$. If it does, abort and restart from (a).
			      \item Finally, apply the same measurement to the query input register. If it returns ``not $\hat{\mathbf a}$'', abort and restart from (a).
		      \end{enumerate}
		\item Run $(x,\pi)\leftarrow\mathcal P^{H''}_{\hat{\mathbf a}}$ from its $i$-th query, starting with the prepared joint state $\ket{\psi_i}\ket{+^{\ell}}^{\otimes ([k]\cdot \challspacesize \cdot |\mathcal Z|)}$ of prover and $H'$ oracle database, where
		      \begin{align*}
			      H''(\mathbf a,i,c,z)=\begin{cases}
				                           H(\mathbf a,i,c,z)  & \text{ if }\mathbf a=\hat{\mathbf a} \\
				                           H'(\mathbf a,i,c,z) & \text{ else.}
			                           \end{cases}
		      \end{align*}
		      Here, it is understood that $H'$ is still instantiated with the compressed oracle simulated by the prover $\mathcal P_{\hat{\mathbf a}, i}^H$, while $H$ is the oracle the prover $\mathcal P_{\hat{\mathbf a}, i}^H$ has access to.
		\item Output $(x,\pi)$.
	\end{enumerate}
	This prover might not be time-efficient. This is not a problem, however, as all extractability statements hold for provers only (at most) fulfilling a query bound.

	We can now relate the probability that $\mathcal P$ succeeds and the extractor fails to the provers $\mathcal P_{\hat{\mathbf a}, i}^H$ as follows,
	\begin{align*}
		    & \Pr_{(x,\pi,w)\leftarrow\langle \mathcal P, \mathcal E\rangle}[\mathcal V(x,\pi)=1\wedge (x,w)\notin R]                                                                                                                                                                \\
		=   & \sum_{\hat{\mathbf a}}\left\|\mathcal E_fP_s\sum_{i=1}^{q'}\left(UO^H\right)^{q'-i+1}\proj{\hat{\mathbf a}}_{X_{\mathcal M}}\left(UO^H(\mathds 1\!-\!\proj{\hat{\mathbf a}}\right)_{X_{\mathcal M}})^{i-1}\ket{\phi_0}_{XYE}\ket{+^{\ell}}^{\otimes \Xi}_{D}\right\|^2 \\
		\le & \sum_{\hat{\mathbf a}}\Bigg(\sum_{i=1}^{q'}\Big\|\mathcal E_fP_s\left(UO^H\right)^{q'-i+1}                                                                                                                                                                             \\
		    & \quad \quad\quad\quad\proj{\hat{\mathbf a}}_{X_{\mathcal M}}\left(UO^H(\mathds 1\!-\!\proj{\hat{\mathbf a}}\right)_{X_{\mathcal M}})^{i-1}\ket{\phi_0}_{XYE}\ket{+^{\ell}}^{\otimes \Xi}_{D}\Big\|\Bigg)^2                                                             \\
		\le & \delta \sum_{\hat{\mathbf a}}\left(\sum_{i=1}^{q'}\left\|\ket{\tilde \psi_i}\right\|\right)^2,
	\end{align*}
	where the first inequality is the triangle inequality and the second inequality results from assumption \eqref{eq:extred-assump1} applied to the provers $\mathcal P_{\hat{\mathbf a}, i}^H$. It remains to bound the sum of norms on the right-hand side. Define
	\begin{align*}
		\varepsilon_{\hat{\mathbf a},i}=\left\|\proj{\hat{\mathbf a}}(UO)^{i-1}\ket{\phi_0}\ket{+^{\ell}}^{\otimes \Xi}_{D}\right\|^2.
	\end{align*}
	We can now bound, using Cauchy-Schwarz in the first step,
	\begin{align*}
		    & \frac 1 {q'}	\left(\sum_{i=1}^{q'}\left\|\ket{\tilde \psi_i}\right\|\right)^2\le\sum_{i=1}^{q'}\left\|\ket{\tilde \psi_i}\right\|^2                                                                           \\
		=   & \sum_{i=1}^{q'}\bra{\phi_0}_{XYE}\bra{+^{\ell}}^{\otimes \Xi}_{D}\left(\left(\mathds 1-\proj{\hat{\mathbf a}}\right)_{X_{\mathcal M}}\left(O^H\right)^\dagger U^\dagger\right)^{i-1}                          \\
		    & \quad\quad\quad\quad\quad\quad \proj{\hat{\mathbf a}}_{X_{\mathcal M}}\left(UO^H\left(\mathds 1-\proj{\hat{\mathbf a}}\right)_{X_{\mathcal M}}\right)^{i-1}\ket{\phi_0}_{XYE}\ket{+^{\ell}}^{\otimes \Xi}_{D} \\
		=   & 1-\bra{\phi_0}_{XYE}\bra{+^{\ell}}^{\otimes \Xi}_{D}\left(\left(\mathds 1-\proj{\hat{\mathbf a}}\right)_{X_{\mathcal M}}\left(O^H\right)^\dagger U^\dagger\right)^{q'}                                        \\
		    & \quad\quad\quad\quad\quad\quad\left(UO^H\left(\mathds 1-\proj{\hat{\mathbf a}}\right)_{X_{\mathcal M}}\right)^{q'}\ket{\phi_0}_{XYE}\ket{+^{\ell}}^{\otimes \Xi}_{D}                                          \\
		\le & 4\sum_{i=1}^{q'}\varepsilon_{\hat{\mathbf a},i}.
	\end{align*}
	Here, we have used the quantum union bound, \cref{lem:q-union}, with projectors $P_i=\left(\left( O^H\right)^\dagger U^\dagger\right)^{i-1}$ $\left(\mathds 1-\proj{\hat{\mathbf a}}\right)(UO)^{i-1}$ in the last line. We thus finally get
	\begin{align*}
		\Pr_{(x,\pi,w)\leftarrow\langle \mathcal P, \mathcal E\rangle}[\mathcal V(x,\pi)
		=1\wedge (x,w)\notin R] & \leq\delta \sum_{\hat{\mathbf a}}\left(\sum_{i=1}^{q'}\left\|\ket{\tilde \psi_i}\right\|\right)^2 \\
		                        & \le4 q'\cdot \delta  \sum_{\hat{\mathbf a}}\sum_{i=1}^{q'}\varepsilon_{\hat{\mathbf a},i}         \\
		                        & =4 {q'}^2\cdot \delta,
	\end{align*}
	where we have used the normalization of $\ket{\phi_0}$ in the last line.
\end{proof}

\subsection{Main Result}
We now state the main result on QROM security of Fischlin's Transform.

\begin{theorem}[Post-Quantum Security of the Fischlin Transform]
	\label{thm:main}
	Let $\Sigma$ be a $\Sigma$-protocol for a relation $R$ with special soundness,
	unique responses, and commitment entropy. Let the parameters $k$, $\ell$, $c$
	of the Fischlin transform $\mathrm{Fis}[\Sigma]$ satisfy $N = c \cdot 2^\ell
		\cdot \log k$, $2^{1/c} \leq k \leq 2^{2^\ell / (256c)}$, $\ell \geq 14$, and
	$\ell = O(\log \log k)$. Then $\mathrm{Fis}[\Sigma]$ is a non-interactive
	zero-knowledge proof of knowledge with straight-line extractability in the
	\textup{QROM}, concretely satisfying:
	\begin{enumerate}
		\item \textup{(Perfect simulation.)} The extractor simulates the
		      quantum-accessible random oracle perfectly, i.e.\ $\varepsilon_{\mathrm{sim}}
			      = 0$.
		\item \textup{(Straight-line extractability.)} For any $q$-query malicious
		      prover $P^*$, the extraction error satisfies
		      \[
			      \varepsilon_{\mathrm{ex}} \;\leq\; q^2 \cdot \mathrm{negl}(k).
		      \]
		\item \textup{(Zero-knowledge.)} $\mathrm{Fis}[\Sigma]$ is computationally
		      indistinguishable from simulated proofs for any quantum polynomial-time
		      distinguisher.
	\end{enumerate}
\end{theorem}

\begin{proof}
	Perfect simulation and straight-line extractability follow from Lemma~\ref{Extractability main bound} and Corollary~\ref{cor:main};
	zero-knowledge from Theorem~\ref{thm:zk} in \cref{sec: zk} in the Appendix.
\end{proof}

\ifsubmission
	\paragraph{Acknowledgements.}
	We have used AI tools for spell checking, grammar and other English-language-related improvements.
\else
	\section{Acknowledgements}
	We have used AI tools for spell checking, grammar and other English-language-related improvements.
	CM thanks Gorjan Alagic for useful discussions. CM and JS acknowledge support by
	the Independent Research Fund Denmark via a DFF Sapere Aude
	grant (IM-3PQC, grant ID 10.46540/2064-00034B). The authors thank the anonymous reviewers of the CRYPTO 2026 program committee for useful feedback.
\fi

\bibliographystyle{alpha}

\begin{thebibliography}{DFMS22b}

	\bibitem[AHJ{\etalchar{+}}23]{AC:AHJMRY23}
	Carlos {Aguilar Melchor}, Andreas H{\"u}lsing, David Joseph, Christian Majenz,
	Eyal Ronen, and Dongze Yue.
	\newblock {SDitH} in the {QROM}.
	\newblock In Jian Guo and Ron Steinfeld, editors, {\em ASIACRYPT~2023,
			Part~VII}, volume 14444 of {\em {LNCS}}, pages 317--350, Guangzhou, China,
	December~4--8, 2023. Springer, Singapore, Singapore.

	\bibitem[BBB{\etalchar{+}}25]{C:BBBDKM25}
	Carsten Baum, Ward Beullens, Lennart Braun, Cyprien {Delpech de Saint Guilhem},
	Michael Kloo{\ss}, Christian Majenz, Shibam Mukherjee, Emmanuela Orsini,
	Sebastian Ramacher, Christian Rechberger, Lawrence Roy, and Peter Scholl.
	\newblock Shorter, tighter, {FAESTer}: Optimizations and improved ({QROM})
	analysis for {VOLE}-in-the-head signatures.
	\newblock In Yael~Tauman Kalai and Seny~F. Kamara, editors, {\em CRYPTO~2025,
			Part~VI}, volume 16005 of {\em {LNCS}}, pages 124--156, Santa Barbara, CA,
	USA, August~17--21, 2025. Springer, Cham, Switzerland.

	\bibitem[BDF{\etalchar{+}}11]{AC:BDFLSZ11}
	Dan Boneh, {\"O}zg{\"u}r Dagdelen, Marc Fischlin, Anja Lehmann, Christian
	Schaffner, and Mark Zhandry.
	\newblock Random oracles in a quantum world.
	\newblock In Dong~Hoon Lee and Xiaoyun Wang, editors, {\em ASIACRYPT~2011},
	volume 7073 of {\em {LNCS}}, pages 41--69, Seoul, South Korea, December~4--8,
	2011. Springer Berlin Heidelberg, Germany.

	\bibitem[CFHL21]{EC:CFHL21}
	Kai-Min Chung, Serge Fehr, Yu-Hsuan Huang, and Tai-Ning Liao.
	\newblock On the compressed-oracle technique, and post-quantum security of
	proofs of sequential work.
	\newblock In Anne Canteaut and Fran\c{c}ois-Xavier Standaert, editors, {\em
			EUROCRYPT~2021, Part~II}, volume 12697 of {\em {LNCS}}, pages 598--629,
	Zagreb, Croatia, October~17--21, 2021. Springer, Cham, Switzerland.

	\bibitem[CMS19]{TCC:ChiManSpo19}
	Alessandro Chiesa, Peter Manohar, and Nicholas Spooner.
	\newblock Succinct arguments in the quantum random oracle model.
	\newblock In Dennis Hofheinz and Alon Rosen, editors, {\em TCC~2019, Part~II},
	volume 11892 of {\em {LNCS}}, pages 1--29, Nuremberg, Germany, December~1--5,
	2019. Springer, Cham, Switzerland.

	\bibitem[DFMS22a]{C:DFMS22}
	Jelle Don, Serge Fehr, Christian Majenz, and Christian Schaffner.
	\newblock Efficient {NIZKs} and signatures from commit-and-open protocols in
	the {QROM}.
	\newblock In Yevgeniy Dodis and Thomas Shrimpton, editors, {\em CRYPTO~2022,
			Part~II}, volume 13508 of {\em {LNCS}}, pages 729--757, Santa Barbara, CA,
	USA, August~15--18, 2022. Springer, Cham, Switzerland.

	\bibitem[DFMS22b]{EC:DFMS22}
	Jelle Don, Serge Fehr, Christian Majenz, and Christian Schaffner.
	\newblock Online-extractability in the quantum random-oracle model.
	\newblock In Orr Dunkelman and Stefan Dziembowski, editors, {\em
			EUROCRYPT~2022, Part~III}, volume 13277 of {\em {LNCS}}, pages 677--706,
	Trondheim, Norway, May~30~--~June~3, 2022. Springer, Cham, Switzerland.

	\bibitem[Fis05]{C:Fischlin05}
	Marc Fischlin.
	\newblock Communication-efficient non-interactive proofs of knowledge with
	online extractors.
	\newblock In Victor Shoup, editor, {\em CRYPTO~2005}, volume 3621 of {\em
			{LNCS}}, pages 152--168, Santa Barbara, CA, USA, August~14--18, 2005.
	Springer Berlin Heidelberg, Germany.

	\bibitem[GHHM21]{AC:GHHM21}
	Alex~B. Grilo, Kathrin H{\"o}velmanns, Andreas H{\"u}lsing, and Christian
	Majenz.
	\newblock Tight adaptive reprogramming in the {QROM}.
	\newblock In Mehdi Tibouchi and Huaxiong Wang, editors, {\em ASIACRYPT~2021,
			Part~I}, volume 13090 of {\em {LNCS}}, pages 637--667, Singapore,
	December~6--10, 2021. Springer, Cham, Switzerland.

	\bibitem[HJMN24]{C:HJMN24}
	Andreas H{\"u}lsing, David Joseph, Christian Majenz, and Anand~Kumar Narayanan.
	\newblock On round elimination for special-sound multi-round identification and
	the generality of the hypercube for {MPCitH}.
	\newblock In Leonid Reyzin and Douglas Stebila, editors, {\em CRYPTO~2024,
			Part~I}, volume 14920 of {\em {LNCS}}, pages 373--408, Santa Barbara, CA,
	USA, August~18--22, 2024. Springer, Cham, Switzerland.

	\bibitem[Ks22]{AC:Konshe22}
	Yashvanth Kondi and {abhi} {shelat}.
	\newblock Improved straight-line extraction in the random oracle model with
	applications to signature aggregation.
	\newblock In Shweta Agrawal and Dongdai Lin, editors, {\em ASIACRYPT~2022,
			Part~II}, volume 13792 of {\em {LNCS}}, pages 279--309, Taipei, Taiwan,
	December~5--9, 2022. Springer, Cham, Switzerland.

	\bibitem[MU05]{MitzenmacherUpfal05}
	Michael Mitzenmacher and Eli Upfal.
	\newblock {\em Probability and Computing: Randomized Algorithms and
		Probabilistic Analysis}.
	\newblock Cambridge University Press, 1st edition, 2005.

	\bibitem[OV22]{o2022quantum}
	Ryan O'Donnell and Ramgopal Venkateswaran.
	\newblock The quantum union bound made easy.
	\newblock In {\em Symposium on Simplicity in Algorithms (SOSA)}, pages
	314--320. SIAM, 2022.

	\bibitem[Pas03]{C:Pass03}
	Rafael Pass.
	\newblock On deniability in the common reference string and random oracle
	model.
	\newblock In Dan Boneh, editor, {\em CRYPTO~2003}, volume 2729 of {\em {LNCS}},
	pages 316--337, Santa Barbara, CA, USA, August~17--21, 2003. Springer Berlin
	Heidelberg, Germany.

	\bibitem[RT25]{C:RotTes25}
	Lior Rotem and Stefano Tessaro.
	\newblock Straight-line knowledge extraction for multi-round protocols.
	\newblock In Yael~Tauman Kalai and Seny~F. Kamara, editors, {\em CRYPTO~2025,
			Part~VII}, volume 16006 of {\em {LNCS}}, pages 95--127, Santa Barbara, CA,
	USA, August~17--21, 2025. Springer, Cham, Switzerland.

	\bibitem[Unr15]{EC:Unruh15}
	Dominique Unruh.
	\newblock Non-interactive zero-knowledge proofs in the quantum random oracle
	model.
	\newblock In Elisabeth Oswald and Marc Fischlin, editors, {\em EUROCRYPT~2015,
			Part~II}, volume 9057 of {\em {LNCS}}, pages 755--784, Sofia, Bulgaria,
	April~26--30, 2015. Springer Berlin Heidelberg, Germany.

	\bibitem[Zha19]{C:Zhandry19}
	Mark Zhandry.
	\newblock How to record quantum queries, and applications to quantum
	indifferentiability.
	\newblock In Alexandra Boldyreva and Daniele Micciancio, editors, {\em
			CRYPTO~2019, Part~II}, volume 11693 of {\em {LNCS}}, pages 239--268, Santa
	Barbara, CA, USA, August~18--22, 2019. Springer, Cham, Switzerland.

\end{thebibliography}
\newcommand{\etalchar}[1]{$^{#1}$}

\appendix
\section*{Appendix}
\section{Deferred proof of \cref{lem:cum-mean-bound}}\label{app:mu'bound}
\begin{proof}[of \cref{lem:cum-mean-bound}]
	Let $\ket{\beta}_{\tilde D_{[i+1:m']}}$ be the post-measurement state after the first $i-1$ subregisters of $\tilde D$ have been measured and fixed outcomes $y_1,\ldots y_{i-1}$ have been obtained. Then $\ket{\beta}\in W^{m-i+1}_{n-i+1}$. When measuring the first remaining register of $\ket{\beta}$, the probability of obtaining outcome $y_{i}=0^{\ell}$ is thus lower-bounded as
	\begin{equation*}
		\Pr[X_{i}=1|y_1,\ldots, y_{i-1}]\ge \frac{2^{-l}(n-i+1)}{(m-i+1)}-2^{1-l / 2} \sqrt{\frac{(n-i+1)(m-n)}{(m-i+1)^2}}
	\end{equation*}
	by \cref{lem:measure}
	and therefore in particular
	\begin{align*}
		\mathbb E[X_{i}|X_{< i}]=\Pr[X_{i}=1|X_{< i}] & \ge  \frac{2^{-l}(n-i+1)}{(m-i+1)}-2^{1-l / 2} \sqrt{\frac{(n-i+1)(m-n)}{(m-i+1)^2}} \\
		                                              & \ge \frac{2^{-l}(n-i)}{(m-i)}-2^{1-l / 2} \sqrt{\frac{(n-i)(m-n)}{(m-i)^2}}.
	\end{align*}
	We can now lower-bound now for any $n'\le n$,
	\begin{align*}
		\mu' & = \sum_{i=1}^{n}\mathbb E[X_{i}|X_{< i}]                                                                     \\
		     & \geq \sum_{i=1}^{n'}\mathbb E[X_{i}|X_{< i}]                                                                 \\
		     & \geq \sum_{i=1}^{n'}\left(\frac{n-i}{m-i}\right) 2^{-l}-2^{1-l/ 2} \sqrt{\frac{(n-i)(m-n)}{(m-i)^2}}         \\
		     & = \sum_{i=1}^{n'}\left(1-\frac{m-n}{m-i}\right) 2^{-l}-2^{1-l / 2} \sqrt{\frac{(n-i)(m-n)}{(m-i)^2}}.        \\
		     & =n'- 2^{-l}(m-n)\sum_{i=1}^{n'}\frac{1}{m-i} - \sum_{i=1}^{n'}2^{1-l / 2} \sqrt{\frac{(n-i)(m-n)}{(m-i)^2}}.
	\end{align*}
	We have the following inequalities for $ n'< m $ on the above terms in the sum:
	\begin{align*}
		\text{1. On the second term:} & \sum_{i=1}^{n'}\frac{1}{m-i}\leq\int_1^{n'}\frac{dx}{m-x}=\log\left(\frac{m-1}{m-n'}\right) \text{,}          \\
		\text{2. On the third term:}  & \sum_{i=1}^{n'}\sqrt{\frac{(n-i)}{(m-i)^2}}\leq\sum_{i=1}^{n'}\sqrt{\frac{(m-i)}{(m-i)^2}}                    \\
		                              & \leq\int_{1}^{n'}\frac{dx}{\sqrt{m-x}}=2\left[\sqrt{\left(m-1\right)}-\sqrt{\left(m-n^{\prime}\right)}\right]
	\end{align*}

	Hence we get:
	\begin{align*}
		 & \mu'\geq 2^{-l}\left[n^{\prime}\!-\!(m-n) \log\left(\frac{m-1}{m-n'}\right)\right] \!-\!2^{2-l / 2} \sqrt{m-n}\left[\sqrt{\left(m-1\right)}\!-\!\sqrt{\left(m-n^{\prime}\right)}\right] \\
		 & \geq 2^{-l}\left[n^{\prime}+(m-n) \log\left(\frac{m-n'}{m-1}\right)\right]-2^{2-l/2}\sqrt{m(m-n)}                                                                                       \\
		 & =2^{-l}\left[k(\challspacesize-1-\gamma)-2^{2+l}k\gamma+(k\gamma) \log\left(\frac{\gamma(1+2^{2+l})}{2^{\ell}\log{k}-1}\right)\right]-2^{2-l/2}\sqrt{m(m-n)}                            \\
		 & \geq2^{-l}k\left[(\challspacesize-(1+\gamma+2^{2+l}\gamma)+(\gamma) \log\left(\frac{4\gamma}{\log{k}}\right)\right]-2^{2-l/2}\sqrt{k(\challspacesize-1)\gamma k}                        \\
		 & \geq2^{-l}k\left[(\challspacesize-1-\gamma\left(1+2^{2+l}+ \log\left(\frac{4\gamma}{\log{k}}\right)\right)-2^{2+l/2}\sqrt{\challspacesize\gamma }\right]                                \\
	\end{align*}

	Here, we have chose $n'$ so that we do not unnecessarily lower-bound non-negative terms by negative ones,

	\begin{align}\label{n'}
		\nonumber         & {\left[\frac{n-i}{m-i}\right]2^{-l}- \frac{2^{1-l} \sqrt{(n-i)(m-n)}}{m-i} \geq 0 } \\
		\nonumber\implies & (n-i) 2^{-l} \geq 2^{1-l / 2} \sqrt{(n-i)(m-n)}                                     \\
		\nonumber\implies & (n-i)^2 2^{-2 l} \geq 2^{2-l}(n-i)(m-n)                                             \\
		\nonumber\implies & (n-i) \geq 2^{2+l}(m-n)                                                             \\
		\implies          & i \leq n-2^{2+l}(m-n) \eqqcolon n'.
	\end{align}
	\qed
\end{proof}
We remark that  the parameter restrictions of Lemma~\ref{Extractability main bound} imply $n' \ge 0$, so the truncation in the proof is non-empty.

\section{A gentle measurement lemma}\label{sec:gentle}

\begin{lemma}[Gentle measurement with additive error]\label{P_n ext norm lemma}
	$||(P_{\eta}\CompOperatorComp{E}_f-\CompOperatorComp{E}_f)|\phi_{\mathrm{succ}}\rangle||^2
		\leq 2{\varepsilon'}+8\sqrt{\varepsilon_{\gamma,(1-\gamma)k}}$.
	\begin{proof}
		Directly expanding:
		$$
			\begin{aligned}
				||(P_{\eta}\CompOperatorComp{E}_f-\CompOperatorComp{E}_f)|\phi_{\mathrm{succ}}\rangle||^2 & =
				|| P_{\eta}\CompOperatorComp{E}_f\ket{\phi_{\mathrm{succ}}}||^2+||\CompOperatorComp{E}_f\ket{\phi_{\mathrm{succ}}}||^2-2 \langle\phi_{\mathrm{succ}}|\CompOperatorComp{E}_f P_{\eta}\CompOperatorComp{E}_f\left|\phi_{\mathrm{succ}}\right\rangle                                                      \\
				                                                                                          & = || P_{\eta}\CompOperatorComp{E}_f\ket{\phi_{\mathrm{succ}}}||^2+||\CompOperatorComp{E}_f\ket{\phi_{\mathrm{succ}}}||^2- 2||P_{\eta}\CompOperatorComp{E}_f  \ket{\phi_{\mathrm{succ}}}||^2                \\
				                                                                                          & \leq  || P_{\eta}||^2||\CompOperatorComp{E}_f\ket{\phi_{\mathrm{succ}}}||^2+||\CompOperatorComp{E}_f\ket{\phi_{\mathrm{succ}}}||^2-2(1-\varepsilon')||\CompOperatorComp{E}_f\ket{\phi_{\mathrm{succ}}}||^2 \\ &\quad +8\sqrt{\varepsilon_{\gamma,(1-\gamma)k}} \text{(due to lemma \ref{lem:ineq 33} )}\\
				                                                                                          & \le (1+1-2(1-\varepsilon'))||\CompOperatorComp{E}_f\ket{\phi_{\mathrm{succ}}}||^2+8\sqrt{\varepsilon_{\gamma,(1-\gamma)k}}                                                                                 \\
				                                                                                          & =2\varepsilon'||\CompOperatorComp{E}_f|\phi_{\mathrm{succ}}\rangle||^2+8\sqrt{\varepsilon_{\gamma,(1-\gamma)k}}                                                                                            \\
				                                                                                          & \leq 2\varepsilon'+8\sqrt{\varepsilon_{\gamma,(1-\gamma)k}}
			\end{aligned}
		$$
		\qed
	\end{proof}
\end{lemma}

\section{Bounding Lemmas }

The following Lemmas gives numerical analysis for the bounds.
\begin{lemma}\label{lem: negligiblity value lemma}
	For some constant $c$, let the challenge space be of size  $\challspacesize=c\cdot 2^{\ell}\cdot \log{(k)}$. Let further $\gamma=4\cdot 2^{-l}$, $2^{\frac 1 c}\le k\le 2^{\frac{2^{\ell}}{256\cdot c}}$ and $l\ge 14$. Then

	$$\varepsilon''+2\sqrt{{\varepsilon'}}+7\sqrt{\sqrt{\varepsilon_{\gamma,(1-\gamma)k}}}\leq  3 \exp\left[-\frac{ k  }{128 \cdot c\cdot 2^{\ell}\cdot \log k}\right]+7\exp\left[-\frac{ k}{8\cdot 2^{\ell}} \right]=\mathrm{negl}(k).$$

\end{lemma}
\begin{proof}

	Substituting the formulas for the parameters we get
	\begin{align}\label{e''+e'+e}
		\nonumber\bar\varepsilon & \coloneqq\varepsilon''+2\sqrt{{\varepsilon'}}+7\sqrt{\sqrt{\varepsilon_{\gamma,(1-\gamma)k}}}                                                                                         \\
		                         & \le\exp\left[-\frac{\delta^2\mu}{3}\right]+2\exp \left[-\frac{\delta'^2\underline{\mu'}}{2^{\ell+1}}\right]+7\exp\left[{-\left(\gamma-2 \cdot 2^{-l }\right) (1-\gamma)k / 8 }\right]\end{align}
	using  \cref{lem:epsilon-prime-value} below.

	We simplify
	\begin{align}\nonumber\bar\varepsilon\le\exp\left[-\frac{\delta^2\mu}{3}\right]+2\exp \left[-\frac{\delta'^2\underline{\mu'}}{2^{\ell+1}}\right]+7\exp\left[- 2^{-l } k / 8 \right]
	\end{align}
	where we have used $1-\gamma\ge 1/2$.

	We bound
	\begin{align}\nonumber
		\underline{\mu'}
		 & =2^{-l}k\left[\challspacesize-1-\gamma\left(1+2^{2+l}+ \log\left(\frac{4\gamma}{\log{k}}\right)\right)-2^{2+l/2}\sqrt{\challspacesize\gamma }\right]                        \\
		 & = \mu -2^{-l}k\left[1+\gamma\left(1+2^{2+l}+ \log\left(\frac{4\gamma}{\log{k}}\right)\right)+2^{2+l/2}\sqrt{\challspacesize\gamma }\right]\leq \mu\label{eq:mu-prime-le-mu}
	\end{align}
	as $\log\left(\frac{4\gamma}{\log k}\right)\ge 4-2l\ge-4\cdot 2^{\ell}$ by the assumptions $k\le 2^{2^{\ell}}$ and $l\ge 14$.

	For simplicity, we set $\delta\mu=\delta'\underline{\mu}'$. We thus get
	\begin{align}
		\delta'^2\underline{\mu}' & =\frac{\delta^2\mu^2}{\underline{\mu}'}\nonumber \\
		\ge \delta^2\mu.\label{eq:complicated-ge-simple}
	\end{align}
	using \cref{eq:mu-prime-le-mu}.

	We have $\eta=\mu(1+\delta)$ (from applying \cref{lem:simple-tail-bound}) and  $ \eta-(1-2 \gamma) k=\underline{\mu^{\prime}}\left(1-\delta^{\prime}\right) $ (see \cref{lem:ineq 33}). We thus get
	\begin{align}
		                & \mu(1+\delta)-(1-2 \gamma) k=\underline{\mu^{\prime}}\left(1-\delta^{\prime}\right)\nonumber \\
		\Leftrightarrow & \delta= \frac{(1-8\cdot 2^{-l})k-(\mu-\underline{\mu}')}{2\mu}.\label{eq:delta}
	\end{align}
	We bound
	\begin{align*}
		\mu-\underline{\mu}' & =2^{-l}k\left[1+\gamma\left(1+2^{2+l}+ \log\left(\frac{4\gamma}{\log{k}}\right)\right)+2^{2+l/2}\sqrt{\challspacesize\gamma }\right] \\
		                     & \le 2^{-l}k\left[1+4\cdot 2^{-l}\left(1+4\cdot 2^{\ell}\right)+8\sqrt{\challspacesize}\right]                                        \\
		                     & \le 2^{-l/2}k\left[21\cdot 2^{-l/2}+8\sqrt{c\cdot \log k}\right]
	\end{align*}
	We can use the assumption on $k$,
	\begin{align*}
		8\sqrt{c\cdot \log k}\le 8\sqrt{\frac{2^{\ell}}{256}}\le \frac 1 2 2^{\ell/2},
	\end{align*}
	so we get
	\begin{align*}
		\mu-\underline{\mu}' & \le  2^{-l/2}k\left[21\cdot 2^{-l/2}+8\sqrt{c\cdot \log k}\right] \\
		                     & \le \frac{3 k}{4},
	\end{align*}
	using $l\ge 14$. Plugging into \cref{eq:delta} we get \begin{align*}
		\delta & \ge \frac{\left(\frac 1 4 -8\cdot 2^{-l}\right)k}{2\mu} \\
		       & \ge \frac{1}{8 c\log k}.
	\end{align*}
	Here we have used $8\cdot 2^{-l}\le 2^{-11}\le 1/8$ due to the assumption $l\ge 14$. We can now use this to get
	\begin{align*}
		\delta^2\mu & \ge \frac{2^{-l} k \challspacesize}{64 c^2\log^2k}                 \\
		            & = \frac{2^{-l} k \cdot c\cdot 2^{\ell}\cdot \log k}{64 c^2\log^2k} \\
		            & =\frac{ k  }{64 c\log k}.
	\end{align*}
	Finally, we obtain
	\begin{align*}\nonumber\bar\varepsilon & \le\exp\left[-\frac{\delta^2\mu}{3}\right]+2\exp \left[-\frac{\delta'^2\underline{\mu'}}{2^{\ell+1}}\right]+7\exp\left[- 2^{-l } k / 8 \right] \\
                                       & \le \exp\left[-\frac{\delta^2\mu}{3}\right]+2\exp \left[-\frac{\delta^2\mu}{2^{\ell+1}}\right]+7\exp\left[- 2^{-l } k / 8 \right]              \\
                                       & \le 3 \exp \left[-\frac{\delta^2\mu}{2^{\ell+1}}\right]+7\exp\left[- 2^{-l } k / 8 \right]                                                     \\
                                       & \le 3 \exp\left[-\frac{ k  }{128 \cdot c\cdot 2^{\ell}\cdot \log k}\right]+7\exp\left[- 2^{-l } k / 8 \right].
	\end{align*}
	\qed
\end{proof}

We prove the deferred lemma.

\begin{lemma}\label{lem:epsilon-prime-value}
	For some constant $c$, let the challenge space be of size  $\challspacesize=c\cdot 2^{\ell}\cdot \log{(k)}$. Let further $\gamma=4\cdot 2^{-l}$, $2^{\frac 1 c}\le k\le 2^{\frac{2^{\ell}}{256\cdot c}}$ and $l\ge 14$. Then
	\begin{align*}
		\exp\left(-\frac{{\delta'}^2\underline\mu'^2}{2 m}\right) & \le\;
		\exp\!\left(
		- \frac{\delta'^2\underline{\mu}'}{2^{\ell}},
		\right) .
	\end{align*}
\end{lemma}

\begin{proof}
	We bound
	\begin{align*}
		\exp\!\left(
		- \frac{\delta'^2\underline{\mu'}^2}{2m}
		\right)
		 & \leq
		\exp\!\left(
		- \frac{\delta'^2\underline{\mu'}^2}{2k\challspacesize}
		\right) \\
		 & \leq
		\exp\!\left(
		- \frac{\delta'^2\underline{\mu'}}{2^{\ell}},
		\right),
	\end{align*}
	where the first inequality is straightforward from $m=k(\challspacesize-1)\le k\challspacesize$, and the second inequality is obtained using the definition of $\underline{\mu'}$ and the assumptions on the parameters, \begin{align*}
		\frac{\underline{\mu'}}{\challspacesize} & =\frac{2^{-l}k\left[\challspacesize-1-\gamma\left(1+4\cdot 2^{\ell}-\log\left(\frac{\log{k}}{4\gamma}\right)\right)-4\sqrt{2^{\ell}\gamma\challspacesize}\right]}{k\challspacesize} \\
		                                         & \geq2^{-l}\left[1-\frac{(1+\gamma\left(1+4\cdot 2^{\ell}\right)+4\sqrt{2^{\ell}\gamma\challspacesize})}{\challspacesize}\right]                                                     \\
		                                         & \geq2^{-l}\left[1-\frac{(1+4\cdot 2^{-l}\left(1+4\cdot 2^{\ell}\right)+4\sqrt{4\challspacesize})}{\challspacesize}\right]                                                           \\
		                                         & \geq2^{-l}\left[1-\frac{(21+8\sqrt{\challspacesize})}{\challspacesize}\right]                                                                                                       \\
		                                         & \geq2^{-l}\left[1-\frac{29}{\sqrt{c\log k\cdot 2^{\ell}}}\right]                                                                                                                    \\
		                                         & \geq2^{-l}\left[1-\frac{29}{\sqrt{ 2^{\ell}}}\right]                                                                                                                                \\
		                                         & \geq2^{-l}\left[1-29\cdot 2^{-7}\right]                                                                                                                                             \\
		                                         & \geq \frac 1 2 2^{-l}
	\end{align*}
	\qed\end{proof}

\section{Proof of Zero Knowledge}\label{sec: zk}

\js{need to change this too as per rebuttal round comment about incremental responses or change incremental response in main fischlin definition to randomised ones} We prove that the Fischlin transform is zero-knowledge in the quantum random oracle model (QROM). Our proof follows standard techniques for adaptive oracle reprogramming in the QROM \cite{AC:GHHM21}. We begin by defining the zero-knowledge simulator. Let $\mathsf{Sim}_\Sigma$ denote the honest-verifier zero-knowledge simulator for the underlying $\Sigma$-protocol as in \cref{SHVZK DEFN} \js{should I change $\mathsf{Sim}_\Sigma$ to $\mathsf{Sim}_{SHVZK}$ ? LOOKS better here with $_{\Sigma}$ reader remembers it's from Sigma Protocol}.

\begin{definition}[Simulator $\mathsf{Sim}_{\mathsf{Fis}}$]\label{def:sim-fis}
	The simulator $\mathsf{Sim}^H_{\mathsf{Fis}}$ proceeds as follows.

	\medskip
	\noindent
	\textbf{Setup.}
	Sample a random function
	\[
		\tilde H : [k]\times \mathcal C \to \{0,1\}^{\ell} .
	\]

	For each $i\in[k]$, sample $c_i\leftarrow^{\$}\mathcal C$ such that
	$\tilde H(i,c_i)=0^{\ell}$.
	If no such value exists, abort.

	\medskip
	\noindent
	\textbf{Proof Generation.}
	For each $i\in[k]$, run
	\[
		(a_i,z_i)\leftarrow \mathsf{Sim}_\Sigma(x,c_i).
	\]
	Set $\mathbf a=(a_i)_{i\in[k]}$ and similarly $\mathbf c$, $\mathbf z$.

	\medskip
	\noindent
	\textbf{Oracle Programming.}
	For any query $(\mathbf a,i,c,z)$, it answers as follows:
	\[
		H'(\mathbf a,i,c,z)=
		\begin{cases}
			\tilde H(i,c)      & \text{if } V(a_i,c,z)=1 \\
			H(\mathbf a,i,c,z) & \text{otherwise}.
		\end{cases}
	\]

	\noindent
	\textbf{Output.}	Define the proof
	\[
		\pi=(\mathbf a,\mathbf c,\mathbf z).
	\]
	and output $(\pi,H')$.

\end{definition}
With this simulator in hand, we can prove zero-knowledge.
\begin{theorem}\label{thm:zk}
	For a \sigp with commitment entropy, the Fischlin transform is zero-knowledge in the QROM.
\end{theorem}

\begin{proof}
	The simulator clearly runs in expected polynomial time.

	\paragraph{Hybrid Argument.}

	We prove indistinguishability by a sequence of hybrid experiments.

	\medskip
	\noindent
	\textbf{Hybrid $H_0$.}
	The honest prover executes the Fischlin protocol using an unmodified quantum random oracle.

	\medskip
	\noindent
	\textbf{Hybrid $H_1$.}
	The prover generates commitments $(a_i)_{i\in[k]}$ honestly.
	The oracle is reprogrammed on all accepting transcripts with prefix $(\mathbf a,i)$ as in Definition~\ref{def:sim-fis}.
	Challenges $c_i$ are sampled uniformly, and responses $z_i$ are computed honestly.

	By the adaptive reprogramming lemma of \cite{AC:GHHM21} (\cref{lem:adarep}), $H_1$ is computationally indistinguishable from $H_0$, except with probability negligible in the min-entropy of the commitment vector.

	\medskip
	\noindent
	\textbf{Hybrid $H_1'$.}
	As in $H_1$, but challenges $c_i$ are sampled conditioned on $\tilde H(i,c_i)=0^{\ell}$.

	Since $\tilde H$ is uniformly random, this conditioning affects the distribution only negligibly. Hence $H_1'$ is indistinguishable from $H_1$.

	\medskip
	\noindent
	\textbf{Hybrid $H_2$.}
	As in $H_1'$, but responses are generated using the simulator $\mathsf{Sim}_\Sigma$ instead of the honest prover.

	By the honest-verifier zero-knowledge property of the $\Sigma$-protocol, $H_2$ is indistinguishable from $H_1'$ up to advantage at most $k\epsilon$, where $\epsilon$ is the distinguishing advantage for a single repetition. \js{need to clarify this epsilon is different from other one in extractability bound}

	Hybrid $H_2$ corresponds exactly to the output of $\mathsf{Sim}_{\mathsf{Fis}}$.
	Since each transition introduces only negligible distinguishing advantage, the simulator output is computationally indistinguishable from a real proof.
	This establishes zero knowledge in the QROM.
	\qed\end{proof}

\end{document}